\documentclass[11pt]{article}
\usepackage{url,amsfonts, amsmath, amssymb, amsthm}
\usepackage{algorithmic,algorithm}
\usepackage{thmtools}
\usepackage{enumitem}
\usepackage[table,xcdraw]{xcolor}
\usepackage{thm-restate}
\usepackage[pdftex]{graphicx}
\usepackage{subcaption}
\usepackage{framed}
\usepackage[framemethod=tikz]{mdframed}
\newcommand{\shortbar}{\begin{center}\rule{5ex}{0.1pt}\end{center}}

\theoremstyle{plain}
\newtheorem{theorem}{Theorem}[section]
\newtheorem{lemma}[theorem]{Lemma}
\newtheorem{corollary}[theorem]{Corollary}

\theoremstyle{definition}

\theoremstyle{remark}

\newcommand{\defeq}{\stackrel{\textrm{def}}{=}}

\newcommand{\E}{{\mathbb E\/}}

\newcommand{\ceil}[1]{\lceil #1 \rceil}

\newcommand{\poly}{\operatorname{poly}}
\newcommand{\rank}{\operatorname{rank}}

\newcommand{\ID}{\operatorname{ID}}

\newcommand{\Compact}{\operatorname{Compact}}

\newcommand{\INDSTATE}[1][1]{\STATE\hspace{#1\algorithmicindent}}

\newcommand{\ignore}[1]{}

\newcommand{\eps}{\epsilon}
\newcommand{\oneeps}{\frac{1}{\epsilon}}

\usepackage{environ}

\newcommand{\repeatlemma}[1]{%
\begingroup
\renewcommand{\thelemma}{\ref{#1}}%
\expandafter\expandafter\expandafter\lemma
\csname replemma@#1\endcsname
\endlemma
\endgroup
}

\NewEnviron{replemma}[1]{%
\global\expandafter\xdef\csname replemma@#1\endcsname{%
\unexpanded\expandafter{\BODY}%
}%
\expandafter\lemma\BODY\unskip\label{#1}\endlemma
}

\newcommand{\repeattheorem}[1]{%
\begingroup
\renewcommand{\thetheorem}{\ref{#1}}%
\expandafter\expandafter\expandafter\theorem
\csname reptheorem@#1\endcsname
\endtheorem
\endgroup
}

\NewEnviron{reptheorem}[1]{%
\global\expandafter\xdef\csname reptheorem@#1\endcsname{%
\unexpanded\expandafter{\BODY}%
}%
\expandafter\theorem\BODY\unskip\label{#1}\endtheorem
}

\begin{document}

\title{Optimal Gossip Algorithms for Exact and Approximate Quantile Computations}
\author{Bernhard Haeupler \\ CMU \and Jeet Mohapatra \\ MIT \and Hsin-Hao Su \\ UNC Charlotte}
\date{}

\maketitle
\begin{abstract}
This paper gives drastically faster gossip algorithms to compute exact and approximate quantiles.

\smallskip

Gossip algorithms, which allow each node to contact a uniformly random other node in each round, have been intensely studied and been adopted in many applications due to their fast convergence and their robustness to failures. Kempe et al.~\cite[FOCS'03]{KDG03} gave gossip algorithms to compute important aggregate statistics if every node is given a value. In particular, they gave a beautiful $O(\log n + \log \oneeps)$ round algorithm to $\eps$-approximate the sum of all values and an $O(\log^2 n)$ round algorithm to compute the exact $\phi$-quantile, i.e., the the $\ceil{\phi n}$ smallest value. 

\smallskip

We give an quadratically faster and in fact optimal gossip algorithm for the exact $\phi$-quantile problem which runs in $O(\log n)$ rounds. We furthermore show that one can achieve an exponential speedup if one allows for an $\eps$-approximation. 
 We give an $O(\log \log n + \log \oneeps)$ round gossip algorithm which computes a value of rank between $\phi n$ and $(\phi+\eps)n$ at every node.
 Our algorithms are extremely simple and very robust - they can be operated with the same running times even if every transmission fails with a, potentially different, constant probability. We also give a matching $\Omega(\log \log n + \log \oneeps)$ lower bound which shows that our algorithm is optimal for all values of $\eps$. 



\end{abstract}
\thispagestyle{empty}
\clearpage
\setcounter{page}{1}

\section{Introduction}
Today, due to the vast amount of data and advances in connectivity between computers, distributed data processing has become increasingly important. In distributed systems, data is stored across different nodes. When one requires some aggregate properties of the data, such as, sums, ranks, quantiles, or other statistics, nodes must communicate in order to compute these properties.  Aggregating data in an efficient and reliable way is a central topic in distributed systems such as P2P networks and sensor networks \cite{YG02, MFHH02, SBAS04}.

\smallskip

We consider uniform gossip protocols, which are often very practical due to their fast convergence, their simplicity, and their stability under stress and disruptions. In uniform gossiping protocols, computation proceeds in synchronized rounds. In each round, each node chooses to {\bf pull} or {\bf push}. In a push a node chooses a message, which is delivered to a uniformly random other node. In a pull each node receives a message from a random node. The message size is typically restricted to $O(\log n)$ bits. The (time) complexity of an algorithm is measured by the number rounds executed. One typically wants algorithms that succeed with high probability, i.e., with probability at least $1 - 1/\poly(n)$.

\smallskip


In this paper, we study the quantile computation problem. In the {\bf exact $\phi$-quantile problem} each node $v$ is given a 
 $O(\log n)$ bit value $x_v$ and wants to compute the $\ceil{\phi n}$ smallest value overall. In the {\bf $\eps$-approximate $\phi$-quantile problem} each node wants to compute a value whose rank is between $(\phi + \eps) n$ and $(\phi - \eps) n$.

%
%
%
%
%

\smallskip

Previously, Kempe et al.~\cite{KDG03} gave a beautiful and simple $O(\log n + \log \oneeps)$ round gossip algorithm to approximate the sum of all values up to a $(1 \pm \epsilon)$ factor. They also showed how to use this algorithm to solve the exact $\phi$-quantile problem in $O(\log^2 n)$ rounds with high probability. 

\smallskip

The main result of this paper is a quadratically faster gossip algorithm for the $\phi$-quantile problem. The $O(\log n)$ round complexity of our algorithm means that the exact $\phi$-quantile can be computed as fast as broadcasting a single message. 

\begin{reptheorem}{firsttheorem}\label{thm:exact}
For any $\phi \in [0,1]$ there is a uniform gossip algorithm which solves the exact $\phi$-quantile problem in $O(\log n)$ rounds with high probability using $O(\log n)$ bit messages.
\end{reptheorem}

Clearly the running time of this algorithm is optimal as $\Omega(\log n)$ rounds are known to be neccesary to even just broadcast to each node the $\phi$-quantile value, after it has been identified. 

\smallskip

Equally interestingly we show that one achieve even faster algorithms if one considers the approximate $\phi$-quantile problem. While a $O(\frac{1}{\eps^2}\log n)$ round algorithm for an $\eps$-approximation follows from simple sampling and a $O(\log n)$ round algorithm computing the median up to a $\pm O(\sqrt{\frac{\log n}{n}})$ (but not general quantiles) was given by Doerr et al.~\cite{DGMSS11}, no approximation algorithm for the quantile problem with a sub-logarithmic round complexity was known prior to this work. We give an $O(\log \log n + \log (1/\eps))$ round algorithm for $\eps$-approximating any $\phi$-quantile which, for arbitrarily good constant approximations, is exponentially faster than our optimal exact algorithm:

\begin{theorem}\label{thm:themain}
For any constant or non-constant $\epsilon(n) > 0$ and any $\phi \in [0,1]$, there exists a uniform gossip algorithm that solves the $\epsilon$-approximate $\phi$ quantile problem in $O(\log \log n + \log \frac{1}{\eps(n)})$ rounds with high probability using $O(\log n)$ bit messages.
\end{theorem}

We also give a $\Omega(\log \oneeps)$ and a $\Omega(\log \log n)$ lower bound for the $\eps$-approximate $\phi$-quantile problem, which shows that our algorithm is optimal for essentially any value of $\eps$. 

\begin{reptheorem}{lowerbound}
For any $\frac{10\log n}{n} < \epsilon < 1/8$ and $\phi \in [0,1]$,  any gossip algorithm that uses less than $\frac{1}{2}\log \log n$ or less than $\log_{4} \frac{8}{\epsilon}$ round fails to solve the $\eps$-approximate $\phi$-quantile problem with probability at least $1/3$. This remains true even for unlimited message sizes. 
\end{reptheorem}

We furthermore show that our algorithms can be made {\bf robust to random failures}, i.e., the same round complexities apply even if nodes fail with some constant probability. We remark that a small caveat of excluding an $\exp(-t)$ fraction of nodes after a runnig time of $t$ rounds is necessary and indeed optimal given that this is the expected fraction of nodes which will not have participated in any successful push or pull after $t$ rounds.

\begin{reptheorem}{robustness}\label{thm:robust}
Suppose in every round every node fails with a, potentially different, probability bounded by some constant $\mu < 1$. For any $\phi \in [0,1]$ there is a gossip algorithm that solves the $\phi$ quantile problem in $O(\log n)$ rounds. For any $t$ and any $\epsilon(n) > 0$, there furthermore exists a gossip algorithm that solves the $\epsilon$-approximate $\phi$-quantile problem in $O(\log \log n + \log \frac{1}{\eps(n)} + t)$ rounds for all but $\frac{n}{2^{t}}$ nodes, with high probability.
\end{reptheorem}

Despite being extremely fast and robust our algorithms remain very simple. In fact they do little more than repeatedly sampling two or three nodes, requesting their current value and selecting the largest, smallest, or median value. Given that in many cases an $\eps$-approximation is more than sufficient we expect this algorithm to easily find applications in areas like sensor networks and distributed database systems. For instance, suppose that a sensor network consisting of thousands of devices is spread across an object to monitor and control the temperature. Say the top and bottom $10\%$-quantiles need special attention. By computing the $90\%$- and $10\%$-quantile, each node can determine whether it lies in the range. It is unlikely that such a computation needs to exact. In fact, for any $0 < \eps < 1$, running $O(1/\eps)$ approximate quantile computations suffices for each node to determine its own quantile/rank up to an additive $\eps$. The fact that this can be done in $O(\log \log n)$ rounds further demonstrates the power of approximations as there is not even a $o(n)$ algorithm known\footnote{The trivial algorithm which broadcasts the maximum value $n$ times requires $O(n\log n)$ rounds. Using network coding gossip~\cite{Haeupler16} one can improve this to $O(n)$ rounds. It is likely that computing the exact rank at each node cannot be done faster.} which allows each node to compute its quantile/rank exactly.

\begin{corollary}\label{thm:themain}
For any constant or non-constant $\eps > 0$ and given that each node has a value there is a gossip algorithm that allows every node to approximate the quantile of its value up to an additive $\epsilon$ in $\oneeps \cdot O(\log \log n + \log \frac{1}{\eps})$ rounds with high probability using $O(\log n)$ bit messages.
\end{corollary}

\paragraph{Technical Summary.} While computing exact quantiles seems to be a very different problem from the approximate problem at first we achieve our exact algorithm by first designing an extremely efficient solution to the $\eps$-approximate quantile problem running in $O(\log \log n + \log \oneeps)$ rounds. This algorithm is in some sense based on sampling and inherently is not able to work for the exact quantile problem or too small $\eps$ itself, as it might randomly discard the target value in its first iteration. However, we show that the algorithm does works for $\eps$ larger than some polynomial in $n$. This allows us to bootstrap this algorithm and repeatedly remove a polynomial fraction of values within $O(\log n)$ rounds until, after a constant number of such iterations, only the target value persists. Overall this leads to an $O(\log \log n + \log \oneeps)$ algorithm for the $\eps$-approximate $\phi$-quantile problem which works for $\eps$ values that can be any function of $n$. This is actually generelizes the $O(\log n)$ algorithm for the exact $\phi$-quantile problem, too, given that the $\frac{1}{2n}$-approximate $\phi$-quantile problem is accurate enough to compute exact ranks and quantiles. 

\smallskip

Next we outline an, indeed very simple, intuition of why a time complexity of $O(\log\log n + \log \oneeps)$ is conceivable for the quantile computation problem: Suppose we sample $\Theta(\log n / \epsilon^2)$ many values uniformly and independently at random. With high probability the $\phi$-quantile of the sampled values is an $\epsilon$-approximation to the $\phi$-quantile in the original data. Since each node can sample $t$ node values (with replacement) in $t$ rounds, one immediately get an $O(\log n / \epsilon^2)$ round algorithm that uses $O(\log n)$ bit messages.

\smallskip

Using larger message sizes, it is possible to reduce the number of rounds. Consider the following {\it doubling algorithm}: Each node $v$ maintains a set $S_v$ such that initially $S_v = \{ x_v \}$. In each round, let $t(v)$ be the node contacted by $v$. Node $v$ updates $S_v$ by setting $S_v \leftarrow S_v \cup S_{t(v)}$. Since the set size essentially doubles each round, after $\log O(\log n /\epsilon^2) = O(\log \log n + \log \oneeps)$ rounds, we have sampled $\Omega(\log n /\epsilon^2)$ values uniformly, albeit not quite independently, at random. A careful analysis shows that indeed a $O(\log \log n + \log \oneeps)$ running time can be achieved using messages of size $\Theta(\log^2 n / \epsilon^2)$ bits. 

\smallskip

Our first approach for reducing the message sizes tried to utilize the quantile approximation sketches from the streaming algorithm community (see related work). We managed to reduce the message complexity to $O(\frac{1}{\epsilon}\cdot \log n \cdot (\log \log n + \log(1/\epsilon)))$ by adopting the compactor ideas from these streaming sketches. A write-up of this might be of independent interest and can be found in the Appendix \ref{apx:sampling}. Unfortunately even if one could losslesly port the state-of-the-art compactors scheme from~\cite{KLL16} into this setting, the $\Omega(\oneeps \log \log (1/\delta))$ for getting  $\eps$-approximate quantile sketches with probability $1-\delta$ suggests that one cannot achieve a $o(\log n \log \log n)$ message size this way, even if $\epsilon$ is a constant. In contrast to most other distributed models of computation gossip algorithms furthermore do not allow to easily shift extra factors in the message size over to the round complexity. In fact, we did not manage to devise any $o(\log n)$ round gossip algorithm based on sketching which adheres to the standard $O(\log n)$ bound on message sizes. 

\smallskip

Instead of sticking to the centralized mentality where each node tries to gather information to compute the answer, 
consider the following  \texttt{3-TOURNAMENT} mechanism. In each iteration, each node uniformly sample three values (in three rounds) and assign its value to the middle one\footnote{Such gossip dynamics have also been analyzed in \cite{DGMSS11} for the setting of adversarial node failures and in \cite{feinerman2017breathe} for the setting of random message corruptions.}. Intuitively nodes close to the median should have a higher probability of surviving and being replicated. Indeed, we show that the number of nodes that have values that are $\eps$ close to the median grows exponentially for the first $O(\log \oneeps)$ iterations. After that, the number of nodes with values more than $\eps n$ away from the median decreases double exponentially. For sufficiently large $\eps$ this simple \texttt{3-TOURNAMENT} algorithm gives an approximation of the the median in $O(\log \oneeps + \log \log n)$ iterations. In general, if we want to approximate the $\phi$-quantile, we shift the $[\phi -\epsilon, \phi+\epsilon]$ quantiles to the quantiles around the median by the following \texttt{2-TOURNAMENT} mechanism. If $\phi < 1/2$, each node samples two values and assign its value to the higher one. The case for $\phi > 1/2$ is symmetric. Intuitively this process makes it more likely for nodes with higher/smaller values to survive. Indeed, we show that with a bit of extra care in the last iterations, $O(\log \frac{1}{\epsilon})$ invocations suffice to shift the values  around the $\phi$-quantile to almost exactly the median, at which point one can apply the median approximation algorithm. 





\smallskip

Finally, our lower bound comes from the fact that if one chooses $\Theta(\epsilon n)$ nodes and either gives them a very large or a very small value, then knowing which of the two cases occurred is crucial for computing any $\eps$-approximate quantiles. However, initially only these $\Theta(\epsilon n)$ nodes have such information. We show that it takes $\Omega(\log \log n + \log \oneeps)$ rounds to spread the information from these nodes to every node, regardless of the message size. 

\paragraph{Related Work} The randomized gossip-based algorithms dates back to Demers et al.~\cite{DGHILSSSD87}. The initial studies are on the spreading of a single message \cite{FG85, Pittel87,KSSV00}, where Karp et al.~\cite{KSSV00} showed that $O(\log n)$ round and $O(n \log \log n)$ total messages is sufficient to spread a single message w.h.p. Kempe et al.~\cite{KDG03} studied gossip-based algorithms for the quantile computation problem as well as other aggregation problems such as computing the sum and the average. Kempe et al.~developed $O(\log n)$ rounds algorithm to compute the sum and average w.h.p. Later, efforts have been made to reduce the total messages to $O(n \log \log n)$ for computing the sum and the average \cite{KDNRS06, CP12}. Using the ability to sample and count, Kempe et al.~implemented the classic randomized selection algorithm \cite{Hoare61, FR75} in $O(\log^2 n)$ rounds. 
Doerr et al.~\cite{DGMSS11} considered gossip algorithms for the problem of achieving a stabilizing consensus algorithm under adversarial node failures. They analyze the median rule, i.e., sample three values and keep the middle value, in this setting and show that $O(\log n)$ rounds suffice to converge to an $\pm O(\sqrt{\frac{\log n}{n}})$-approximate median even if $O(\sqrt{n})$ adversarial node failures occur. Similar gossip dynamics were also studied in \cite{feinerman2017breathe} which considers randomly  corrupted (binary) messages.

The exact quantile computation is also known as the selection problem, where the goal is to select the $k$'th smallest element. The problem has been studied extensively in both centralized and distributed settings. Blum et al.~\cite{BFPRT73} gave a deterministic linear time algorithm for the problem in the centralized setting. In the distributed setting, Kuhn et al.~\cite{KLW07} gave an optimal algorithm for the selection problem that runs in $O(D \log_{D} n)$ rounds in the \textsf{CONGEST} model, where $D$ is the diameter of the graph. Many works have been focused on the communication complexity aspect (i.e. the total message size sent by each node) of the problem \cite{SFR83, RSS86, SSS88, SS89, SSS92, NSU97, GK04, Patt-Sharmir07}. Most of them are for complete graphs or stars.  Others studied specific class of graphs such as two nodes connected by an edge \cite{Rodeh82, CT87}, rings, meshes, and complete binary trees \cite{Frederickson83}.

The quantile computation problem has also been studied extensively in the streaming algorithm literature {\cite{MP80,ARS97, MRL99, GK01,GM08, HT10, Agarwal13, MMS13, FO15, KLL16}, where the goal is to approximate $\phi$-quantile using a small space complexity when the data comes in a single stream.

\section{The Tournament Algorithms}

In this section, we present our algorithm for the $\epsilon$-approximate quantile computation problem for sufficiently large $\epsilon$. For convenience, we use $a\pm b$ to denote the interval $[a-b, a+b]$.

\begin{theorem}\label{thm:quantile_approx} For any constant or non-constant $\epsilon(n) = \Omega(1/n^{0.096})$ and any $\phi \in [0,1]$, there exists a uniform gossip algorithm that solves the $\epsilon$-approximate $\phi$ quantile problem in $O(\log \log n + \log \frac{1}{\eps(n)})$ rounds with high probability using $O(\log n)$ bit messages.\end{theorem}


The algorithm is divided into two phases. In the first phase, each node adjusts its value so that the quantiles around the $\phi$-quantile will become the median quantiles approximately. In the second phase, we show how to compute the approximate median.

\subsection{Phase I: Shifting the Target Quantiles to Approximate Medians}\label{section:first_phase}
Let $H$ denote the nodes whose quantiles lie in $(\phi + \epsilon, 1]$, $L$ denote the nodes whose quantiles lie in $[0, \phi - \epsilon)$, and $M$ denote the nodes whose quantiles lie in $[\phi - \epsilon, \phi+\epsilon]$. Let $L_i$, $M_i$, and $H_i$ denote the set $L$, $M$, and $H$ respectively at the end of iteration $i$ for $i \geq 1$ and $L_0$, $M_0$, and $H_0$ denote the sets in the beginning. The goal is to show that by the end of the algorithm at iteration $t$, the size of $|L_t|$ and $|H_t|$ are $1/2 - \Omega(\epsilon)$ so that an approximate median lies in $M_t$, which consists of our target quantiles.

Let $h_0 = 1 - (\phi + \epsilon)$ and $l_0 = \phi - \epsilon$. We first consider the case where $h_0 \geq l_0$ and the other case is symmetric. Initially, we have $\frac{|H_0|}{n} \in h_0 \pm 1/n$.  Let $h_{i+1} = h_{i}^2$ for $i \geq 1$.  We run the \texttt{2-TOURNAMENT} algorithm (Algorithm \ref{algo:shift}) until the iteration $t$ such that $h_t\leq T$, where $T \defeq 1/2 -\epsilon$. We will show that $t = O(\log(1/\epsilon))$ and $\frac{|H_i|}{n}$ concentrates around $h_i$  for iteration $1\leq i \leq t-1$ and in the end we have $\frac{|H_t|}{n} \in  T \pm \frac{\epsilon}{2}$. 

\begin{algorithm}[H]
\begin{algorithmic}[1]
\STATE $h_0 \leftarrow (1-(\phi+\epsilon))$ \\
\STATE $i \leftarrow 0$, $T = 1/2 - \epsilon$.
\WHILE{$h_i > T$}
	\STATE $h_{i+1} \leftarrow h_i^2$ \\
	\STATE $\delta \leftarrow \min\left(1, \frac{h_i - T}{h_{i} - h_{i+1}}\right)$ \\
	\STATE With probability $\delta$ {\bf do}
		\INDSTATE \label{line:2tour1}Select two nodes $t_1(v)$ and $t_2(v)$ randomly 
		\INDSTATE \label{line:2tour2} $x_v \leftarrow \min(x_{t_1(v)}, x_{t_2(v)})$
	\STATE Otherwise {\bf do}
			\INDSTATE \label{line:1tour1} Select a node $t_1(v)$ randomly 
			\INDSTATE \label{line:1tour2} $x_v \leftarrow x_{t_1(v)}$
	\STATE $i \leftarrow i + 1$
\ENDWHILE
\end{algorithmic} 
\caption{\texttt{2-TOURNAMENT}($v$)}
\label{algo:shift}
\end{algorithm}

The algorithm ends when $h_i$ decreases below $T$. The lemma below bounds the number of iterations needed for this to happen. It can be shown that since $h^2_i$ squares in each iteration, the quantity $(1-h_i)$ roughly grows by a constant factor in each iteration. Since initially $(1-h_0) \geq \epsilon$, $O(\log(1/\epsilon))$ iterations suffice for $h_i$ to decrease below $T$. 

\begin{replemma}{shift} Let $t$ denote the number of iterations in Algorithm \ref{algo:shift}, $t \leq \log_{7/4}(4/\epsilon) + 2$. \end{replemma}
\begin{proof}
The algorithm ends when $h_t \leq 1/2 - \epsilon$. The highest possible value for $h_0$ is $1 - \epsilon$. We show that $h_i \leq 1-\left(\frac{7}{4} \right)^{i}\cdot \epsilon$ provided that $1-\left(\frac{7}{4} \right)^{i-1}\cdot \epsilon  \geq 3/4$.

Suppose by induction that $h_{i-1} \leq 1 - \left(\frac{7}{4} \right)^{i-1}\cdot \epsilon$. We have \begin{align*}h_{i} &= h_{i-1}^2 \\
&\leq \left(1 - \left(\frac{7}{4} \right)^{i-1} \cdot \epsilon\right)^2 \\
&= 1 - \left(\frac{7}{4}\right)^{i-1} \cdot \epsilon \cdot \left(2 - \left(\frac{7}{4}\right)^{i-1} \cdot \epsilon \right) \\
&\leq 1 - \left(\frac{7}{4}\right)^{i-1} \cdot \epsilon \cdot \left(2 - \frac{1}{4} \right) 
\leq 1- \left(\frac{7}{4} \right)^i\cdot \epsilon && 1-\left(\frac{7}{4} \right)^{i-1}\cdot \epsilon  \geq 3/4
\end{align*}
Therefore, after $i_0 = \log_{7/4} (4/\epsilon)$ iterations we have, $h_{i_0} \leq 1 -(\frac{7}{4})^{i} \cdot \epsilon \leq 1 -\frac{1}{4} = \frac{3}{4}$. Since $h_{i_0+2} \leq (\frac{3}{4})^4 < \frac{1}{2} - \epsilon$ for $\epsilon < 1/8$, it must be the case $t \leq i_0 + 2 = O(\log(1/\epsilon))$.
\end{proof}

 Note that Line \ref{line:2tour1} and Line \ref{line:2tour2} are executed with probability 1 during the first $t-1$ iterations. We show the following.

\begin{lemma}\label{lem:first_rounds} For iteration $1 \leq i < t - 1$, $\E[\frac{|H_{i+1}|}{n} \mid H_{i}] = \frac{|H_i|^2}{n^2}$. \end{lemma}
\begin{proof}
	$$\E\left[\frac{|H_{i+1}|}{n} \mid H_{i}\right] = \frac{1}{n} \sum_{v \in V} \Pr(x_{t_1(v)} \in H_{i} \wedge x_{t_2(v)} \in H_{i} ) = \frac{1}{n} \sum_{v\in V} \left(\frac{|H_i|}{n}\right)^2 = \left( \frac{|H_i|^2}{n^2} \right) \qedhere $$
\end{proof}

Therefore, for $1 \leq i \leq t-1$, if $\frac{|H_i|}{n} \sim h_i$ ($\sim$ means they are close), then $\frac{|H_{i+1}|}{n} \sim h_{i+1}$ if $\frac{|H_{i+1}|}{n}$ concentrates around its expectation, since $h_{i+1} = h^2_{i}$.

In the last iteration, we truncate the probability of doing the tournament by doing it with only probability $\delta$ for each node, so that ideally we hope to have $\E[\frac{|H_t|}{n} | H_{t-1}] = T$. However, since $\delta$ is calculated with respect to $h_{t-1}$ instead of the actual $\frac{|H_{t-1}|}{n}$, we need the following lemma to bound the deviated expectation. In the next lemma, we show that if $\frac{|H_{t-1}|}{n} \sim h_{t-1}$, then indeed we have $\E[\frac{|H_t|}{n} | H_{t-1}] \sim T$. 

\begin{replemma}{last_round}\label{lem:last_round} Suppose that $(1-\epsilon'')h_{t-1} \leq \frac{|H_{t-1}|}{n} \leq (1+\epsilon'') h_{t-1}$ for some $0<\epsilon''<1$. We have $ T - 3\epsilon'' \leq \E\left[\frac{|H_{t}|}{n} \mid H_{t-1} \right]\leq  T  + 3\epsilon''$.
\end{replemma}
\begin{proof}
First, 

\begin{align*}
\E\left[\frac{|H_{t}|}{n} \mid H_{t-1}\right] &= \frac{1}{n}\sum_{v \in V} \left(\delta \cdot \frac{|H_{t-1}|^2}{n^2} + (1-\delta)\cdot \frac{|H_{t-1}|}{n} \right) \\
&= \frac{|H_{t-1}|}{n} \cdot \left((1-\delta) + \delta \cdot \frac{|H_{t-1}|}{n} \right)
\end{align*}
Suppose that $\frac{|H_{t-1}|}{n} \leq (1+\epsilon'') h_{t-1}$. We have
\begin{align*}
\E\left[\frac{|H_{t}|}{n} \mid H_{t-1}\right] &\leq (1+\epsilon'') h_{t-1} \cdot \left((1-\delta) + \delta \cdot (1+\epsilon'') h_{t-1} \right) \\
&=(1+\epsilon'') h_{t-1} \cdot \left(\frac{T\cdot(1-h_{t-1}) - \epsilon'' h_{t-1}T + \epsilon'' h_{t-1}^2}{h_{t-1} - h_t} \right) \\
&=(1+\epsilon'') \cdot \left(\frac{T\cdot(h_{t-1}-h^2_{t-1})}{h_{t-1} - h_t}  + h_{t-1} \cdot \frac{\epsilon'' h_{t-1}^2 - \epsilon'' h_{t-1}T}{h_{t-1} - h_t} \right) \\
&= (1+\epsilon'') \cdot \left(T  + \epsilon'' h^2_{t-1} \cdot \frac{(h_{t-1} - T)}{h_{t-1} - h_t} \right) \\
&\leq (1+\epsilon'') \cdot \left(T  + \epsilon'' h_{t} \right) \leq T+3\epsilon''
\end{align*}

Similarly, if $\frac{|H_{t-1}|}{n} \geq (1-\epsilon'') h_{t-1}$, we have 
\begin{align*}
\E\left[\frac{|H_{t}|}{n} \mid H_{t-1}\right] &\geq (1-\epsilon'') h_{t-1} \cdot \left((1-\delta) + \delta \cdot (1-\epsilon'') h_{t-1} \right) \\
&=(1-\epsilon'') h_{t-1} \cdot \left(\frac{T\cdot(1-h_{t-1}) + \epsilon'' h_{t-1}T - \epsilon'' h_{t-1}^2}{h_{t-1} - h_t} \right) \\
&=(1-\epsilon'') \cdot \left(\frac{T\cdot(h_{t-1}-h^2_{t-1})}{h_{t-1} - h_t}  - h_{t-1} \cdot \frac{\epsilon'' h_{t-1}^2 - \epsilon'' h_{t-1}T}{h_{t-1} - h_t} \right) \\
&= (1-\epsilon'') \cdot \left(T  - \epsilon'' h^2_{t-1} \cdot \frac{(h_{t-1} - T)}{h_{t-1} - h_t} \right) \\
&\geq (1-\epsilon'') \cdot \left(T  - \epsilon'' h_{t} \right) \geq T-3\epsilon''				 \qedhere
\end{align*}
\end{proof}

In the end, we hope that $\frac{|H_{t}|}{n}$ deviates from $T$ by at most $\epsilon/2$. To achieve this, we show that in each iteration $1 \leq i \leq t-1$, $\frac{|H_{i}|}{n}$ deviates from its expectation, $\frac{|H_{i-1}|^2}{n^2}$, by at most a $(1\pm \epsilon')$ factor, where $\epsilon' = \epsilon/2^{t+4}$ is an error control parameter that is much less than $\epsilon$. Note that $\frac{|H_{i-1}|^2}{n^2}$ is already deviated from $h^2_{i-1} = h_i$ to some degree. The next lemma bounds the cumulative deviation of $\frac{|H_i|}{n}$ from $h_i$. Note that $\epsilon'$ has to be large enough in order guarantee that $\frac{|H_{i}|}{n}$ lies in the $(1\pm \epsilon') \cdot \frac{|H_{i-1}|^2}{n^2}$ range. This also implies $\epsilon$ has to be large enough. 

\begin{replemma}{first_phase_concentrate}\label{lem:first_phase_concentrate} Let $\epsilon = \Omega(1/n^{1/4.47})$ and $\epsilon' = \frac{\epsilon}{2^{t+4}}$. W.h.p.~for iteration $0 \leq i < t$, we have $\frac{|H_i|}{n} \in (1 \pm \epsilon')^{2^{i+1}-1} \cdot h_i$. 
\end{replemma}

\begin{proof}
We will show by induction. Initially, we have $$(1-\epsilon')h_0 \leq (1-1/n)\cdot h_0\leq \frac{|H_0|}{n} \leq (1+1/n)h_0 \leq (1+\epsilon') h_0$$

Suppose that $\frac{|H_i|}{n} \in (1 \pm \epsilon)^{2^{i+1}-1} \cdot h_i$ is true. Then by Lemma \ref{lem:first_rounds}, 
\begin{align}
\E[|H_{i+1}| \mid H_{i}] &= \frac{|H_i|^2}{n^2} \nonumber \\
&\geq (1 - \epsilon')^{2^{i+2}-2} \cdot h^2_i \nonumber  \\
&\geq (1-\epsilon' \cdot (2^{t+2} -2)) \cdot h^2_i && (1-a)^b \geq 1-ab\nonumber \\
&\geq (1-\epsilon) \cdot h^2_i = \Omega(1) && \epsilon' = \frac{\epsilon}{2^{t+4}}, h_i \geq 1/2-\epsilon \label{eqn:1}
\end{align}

Also, \begin{equation}\label{eqn:2}\epsilon'  = \epsilon \cdot 2^{\log_{7/4}(\epsilon/4) - 2}/16\geq \epsilon^{2.24}/256=\Omega(\sqrt{\log n / n})\end{equation}

Therefore, by (\ref{eqn:1}), (\ref{eqn:2}), and Chernoff Bound, we have $$\Pr\left(\frac{|H_{i+1}|}{n} \in (1\pm\epsilon')\cdot \E[\frac{|H_{i+1}|}{n}\mid H_{i}]\right) \leq 1-2\cdot e^{-\Omega(\epsilon'^2 \E[|H_{i+1}| \mid H_i])} \leq 1 - 2\cdot e^{-\Omega(\frac{\log n}{n} \cdot n)} = 1-1/\poly(n) $$

Thus, w.h.p.~$\frac{|H_{i+1}|}{n} \in (1\pm\epsilon')\cdot \E[\frac{|H_{i+1}|}{n}\mid H_{i}]$.
\begin{align*}
\frac{|H_{i+1}|}{n} &\leq (1+\epsilon') \cdot \left(\frac{|H_i|}{n}\right)^2 \\
&\leq  (1+\epsilon') \cdot \left((1+\epsilon')^{2^{i+1}-1} \cdot h_i\right)^2 \\
&= (1+\epsilon')^{2^{i+2}-1} \cdot h_{i+1}
\end{align*}
Similarly, we can show $\frac{|H_i|}{n} \geq  (1-\epsilon')^{2^{i+2}-1} \cdot h_{i+1}$. Therefore, we conclude for $0 \leq i < t$, $\frac{|H_i|}{n} \in (1 \pm \epsilon')^{2^{i+1}-1} \cdot h_i$.
\end{proof}

Combining Lemma \ref{lem:last_round} for the deviation on the last round and Lemma \ref{lem:first_phase_concentrate} for the deviation on first $t-1$ rounds, the following lemma summarizes the final deviation. 
\begin{replemma}{end_H_bound} Let $\epsilon = \Omega(1/n^{1/4.47})$. At the end of the algorithm, w.h.p. $T - \frac{\epsilon}{2} \leq \frac{|H_t|}{n} \leq T + \frac{\epsilon}{2} $. \end{replemma}

\begin{proof}

Let $\epsilon' = \frac{\epsilon}{2^{t+4}}$. By Lemma \ref{lem:first_phase_concentrate}, we have $\frac{|H_{t-1}|}{n} \in (1 \pm \epsilon')^{2^{t} - 1} \cdot h_{t-1}$ w.h.p.

Note that $(1+\epsilon')^{2^{t} - 1} \leq 1+2({2^{t} - 1})\epsilon'\leq 1+2^{t+1}\epsilon'$, since $(1+a)^b \leq 1+2ab$ provided $0 < ab \leq 1/2$ and we have $({2^{t} - 1})\epsilon' \leq \epsilon/16 \leq 1/2$.
Similarly, $(1-\epsilon')^{2^{t} - 1} \geq 1-2^{t+1}\epsilon'$. Therefore,  we can let $\epsilon'' = 2^{t+1}\epsilon'$ and apply Lemma \ref{lem:last_round} to conclude that $$T - 3 \cdot 2^{t+1}\epsilon' \leq T - 3\epsilon'' \leq \E\left[\frac{|H_{t}|}{n} \mid H_{t-1}\right]\leq T - 3\epsilon''\leq  T  +  3\cdot 2^{t+1}\epsilon'$$

By Lemma \ref{lem:first_phase_concentrate}, (\ref{eqn:1}), and (\ref{eqn:2}), we have $\frac{|H_{t-1}|}{n} = \Omega(1)$ and $\epsilon' = \Omega(\sqrt{\log n / n})$. Therefore, we can apply Chernoff Bound to show w.h.p.~$|H_t| \in (1\pm \epsilon)\E[|H_t| \mid H_{t-1}]$. Thus,
\begin{align*}
(1-\epsilon')(T - 3 \cdot 2^{t+1}\epsilon' ) &\leq \frac{|H_{t}|}{n} \leq  (1+\epsilon')(T  +  3 \cdot 2^{t+1}\epsilon' ) \\
T - 4 \cdot 2^{t+1}\epsilon' &\leq \frac{|H_{t}|}{n} \leq  T  +   4 \cdot 2^{t+1}\epsilon' \\
T - \frac{\epsilon}{2} &\leq \frac{|H_{t}|}{n} \leq T + \frac{\epsilon}{2} && \epsilon' =  \epsilon/ (2^{t+4}) \qedhere
\end{align*}
\end{proof}

Initially, $\frac{|M_0|}{n} \geq 2\epsilon$. The following three lemmas are for showing Lemma $\ref{lem:M_lower_bound}$, which states that $\frac{|M_t|}{n}$ does not decrease much from $2\epsilon$. This is done by showing that $\frac{|M_i|}{n}$ does not decrease in each round, except possibly the last. We first derive bounds on the expectation of $\frac{|M_{i+1}|}{n}$ in terms $\frac{|M_i|}{n}$ in each round.

\begin{lemma}\label{lem:M_first_rounds} For $0 \leq i < t-1$, suppose that $\frac{|H_i|}{n} \geq 1/2$ and  $\frac{|M_i|}{n} \geq \frac{|M_0|}{n}$, $\E[\frac{|M_{i+1}|}{n} \mid M_i, H_i] \geq \frac{|M_{i}|}{n}\cdot (1 + 2\epsilon)$.\end{lemma}
\begin{proof}
 Therefore, we have
\begin{align*}
\frac{\E[|M_{i+1}| \mid M_i, H_i]}{n} &= \frac{1}{n} \cdot \sum_{v \in V} \left( 2 \cdot \frac{|M_i|}{n} \cdot \frac{|H_i|}{n} + \left(\frac{|M_i|}{n} \right)^2 \right) \\
&= \frac{1}{n} \cdot \sum_{v \in V} \left( \frac{|M_i|}{n} \left(2 \cdot \frac{|H_i|}{n} + \frac{|M_i|}{n}  \right) \right) \\
&\geq \frac{1}{n} \cdot \sum_{v \in V} \frac{|M_i|}{n}\cdot (1+2\epsilon) && \frac{|H_i|}{n} \geq 1/2 , |M_i| \geq |M_0| \geq 2\epsilon n \\
&\geq \frac{|M_i|}{n} \cdot (1+2\epsilon) \qedhere
\end{align*}
\end{proof}

\begin{lemma}\label{lem:M_last_round} Suppose that $\frac{|H_{t-1}|}{n} \geq \frac{1}{2}-\frac{3\epsilon}{2}$ and $\frac{|M_{t-1}|}{n} \geq \frac{|M_0|}{n}$, $\E[\frac{|M_{t}|}{n} \mid M_{t-1}, H_{t-1}] \geq (1-\epsilon)\cdot\frac{|M_{t-1}|}{n}$.\end{lemma}
\begin{proof}
\begin{align*}
\frac{\E[|M_{t}| \mid M_{t-1}, H_{t-1}]}{n} &= \frac{1}{n} \cdot \sum_{v \in V} \left( \delta \cdot \left( 2 \cdot \frac{|M_{t-1}|}{n} \cdot \frac{|H_{t-1}|}{n} + \left(\frac{|M_{t-1}|}{n} \right)^2\right) + (1-\delta) \cdot \frac{|M_{t-1}|}{n} \right) \\
&= \delta \cdot  \left( \frac{|M_{t-1}|}{n} \left(2 \cdot \frac{|H_{t-1}|}{n} + \frac{|M_{t-1}|}{n}  \right) \right) + (1-\delta)\cdot \frac{|M_{t-1}|}{n} \\
&\geq \delta \cdot \frac{|M_{t-1}|}{n}\cdot (1-\epsilon) + (1-\delta) \cdot \frac{|M_{t-1}|}{n} \hspace{8mm}   \frac{1}{2}-\frac{3\epsilon}{2}, \frac{|M_{t-1}|}{n} \geq \frac{|M_0|}{n} \geq 2\epsilon  \\
&\geq (1-\epsilon)\cdot \frac{|M_{t-1}|}{n}
\qedhere
\end{align*}
\end{proof}

\begin{lemma}\label{lem:M_concentrate} For $0 \leq i \leq t-1$, suppose that $\frac{|M_i|}{n} \geq \frac{|M_0|}{n}$, $\Pr(\frac{|M_{i+1}|}{n} \in (1\pm\epsilon')\cdot \E[\frac{|M_{i+1}|}{n}\mid M_{i}, H_{i}] \leq 1-2\cdot e^{-\Omega(\epsilon'^2 \E[|M_{i+1}| \mid M_i, H_i])}$. \end{lemma}
\begin{proof}
Since the events that $v \in M_{i+1}$ are independent for all $v$,  by Chernoff Bound, 
\begin{align*}
\Pr\left(\frac{|M_{i+1}|}{n} > (1+\epsilon')\cdot \frac{\E[|M_{i+1}| \mid M_i,H_i]}{n}\right) &\leq e^{-\Omega(\epsilon'^2 \E[|M_{i+1} |\mid M_i, H_{i}])} 
\\
\Pr\left(\frac{|M_{i+1}|}{n} < (1-\epsilon')\cdot \frac{\E[|M_{i+1}| \mid M_i, H_i]}{n} \right) &\leq e^{-\Omega(\epsilon'^2 \E[|M_{i+1} |\mid M_i, H_{i}])} \qedhere 
\end{align*}
\end{proof}

\begin{replemma}{M_lower_bound}\label{lem:M_lower_bound} Let $\epsilon = \Omega(1/n^{1/4.47})$. At the end of the algorithm, w.h.p.~$\frac{|M_t|}{n} \geq \frac{7\epsilon}{4}$. \end{replemma}
\begin{proof}
First note that for $0 \leq i \leq t-2$, w.h.p.~we must have $\frac{|H_i|}{n} \geq 1/2$. Otherwise, by Lemma \ref{lem:first_rounds} and Chernoff Bound with error $\epsilon'=1/3$, w.h.p.~we have $\frac{|H_{i+1}|}{n} \leq (1+\epsilon')(1/2)^2 \leq 1/2 - \epsilon$. This implies $t=i+1$, a contradiction occurs.

We will show by induction that $\frac{|M_i|}{n} \geq \frac{|M_0|}{n}$ for $0 \leq i \leq t-1$ first. Initially, $\frac{|M_0|}{n} \geq \frac{|M_0|}{n}$ is true. Suppose that $\frac{|M_i|}{n} \geq \frac{|M_0|}{n}$ is true.

Let $\epsilon' = \epsilon / 4$. By Lemma \ref{lem:M_concentrate} and Lemma \ref{lem:M_first_rounds}, we have 
$\Pr\left(\frac{|M_{i+1}|}{n} < (1-\epsilon')\cdot \frac{\E[|M_{i+1}| \mid M_i, H_i]}{n} \right) \leq e^{-\Omega(\epsilon'^2 \E[|M_{i+1} |\mid M_i, H_{i}])} =  \leq e^{-\Omega(\epsilon'^2 \epsilon n)} = 1/\poly(n)$. Therefore, w.h.p.~$\frac{|M_{i+1}|}{n} \geq (1-\epsilon')\cdot \frac{\E[|M_{i+1}| \mid M_i, H_i]}{n} \geq (1-\epsilon/4)\cdot(1 + 2\epsilon)\cdot \frac{|M_i|}{n} \geq \frac{|M_{i}|}{n}$. From this statement, we can also obtain that in the end of round $t-1$, $$\frac{|M_{t-1}|}{n} \geq (1 - \frac{\epsilon}{4})\cdot (1+2\epsilon)\cdot \frac{|M_{t-2}|}{n} \geq  (1 - \frac{\epsilon}{4})\cdot (1+2\epsilon)\cdot \frac{|M_{0}|}{n} $$

Also let $\epsilon'' = \epsilon/2^{t+4}$. By Lemma \ref{lem:first_phase_concentrate}, w.h.p.~we have $$\frac{|H_{t-1}|}{n} \geq (1-\epsilon'')^{2^{i+1} - 1} \cdot h_{t-1} \geq  (1- 2\cdot (2^{i+1} - 1)\epsilon'')\cdot h_{t-1} \geq   h_{t-1} \cdot (1-\epsilon/2) \geq  T -\frac{\epsilon}{2} \geq \frac{1}{2} - \frac{3\epsilon}{2}$$ 

This satisfy the conditions required for Lemma \ref{lem:M_last_round}. By Lemma \ref{lem:M_concentrate} and Lemma \ref{lem:M_last_round}, we have w.h.p.~
\begin{align*}
\frac{|M_{t}|}{n} &\geq (1-\epsilon')\cdot \frac{\E[|M_{t}| \mid M_{t-1}, H_{t-1}]}{n} \\
& \geq (1-\epsilon')\cdot (1 - \epsilon)\cdot \frac{|M_{t-1}|}{n} \\
& \geq (1-\epsilon/4)^2 \cdot (1+\epsilon) \cdot \frac{|M_{0}|}{n} \\
& \geq (1-\epsilon) \cdot (1+\epsilon) \cdot \frac{|M_{0}|}{n} && \mbox{$(1-\epsilon/4)^2 \geq 1-\epsilon$ for $\epsilon \leq 1/2$.}\\
& \geq (1-\epsilon^2) \cdot \frac{|M_{0}|}{n} \\
& \geq \frac{7\epsilon}{4} && \mbox{$\epsilon < 1/8$ and $\frac{|M_0|}{n} = 2\epsilon$} \qedhere
\end{align*}
\end{proof}

The following lemma shows that at the end of Algorithm \ref{algo:shift}, the problem has been reduced to finding the approximate median.

\begin{lemma}\label{lem:last} Let $\epsilon = \Omega(1/n^{1/4.47})$. At the end of iteration $t$ of Algorithm \ref{algo:shift}, w.h.p.~any $\phi'$-quantile where $\phi' \in [\frac{1}{2} -\frac{\epsilon}{4}, \frac{1}{2} + \frac{\epsilon}{4}]$ must be in $M_t$. \end{lemma}

\begin{proof}


Since $\frac{1}{2} - \frac{3\epsilon}{2} \leq \frac{|H_t|}{n} \leq \frac{1}{2} - \frac{\epsilon}{2}$ and $\frac{|M_t|}{n} \geq \frac{7\epsilon}{4}$ by Lemma \ref{first_phase_concentrate} and Lemma \ref{lem:M_lower_bound}, we have $\frac{|M_t|}{n} +  \frac{|H_t|}{n} \geq \frac{1}{2} + \frac{\epsilon}{4}$. Combined with the fact that $\frac{|H_t|}{n} \leq \frac{1}{2} - \frac{\epsilon}{2}$, we conclude that any $\phi'$-quantile where $\phi' \in [\frac{1}{2} - \frac{\epsilon}{4}, \frac{1}{2} + \frac{\epsilon}{4}]$ must  be in $M_t$. \end{proof}

\subsection{Phase II: Approximating the Median}\label{section:second_phase}
Let $L_i, M_i$, and $H_i$ denote the nodes whose quantiles lie in $[0, \frac{1}{2} - \epsilon)$, $[\frac{1}{2} - \epsilon, \frac{1}{2} + \epsilon]$, and $(\frac{1}{2} - \epsilon, 1]$ respectively at the end of iteration $i$, and $L_0, M_0$, and $H_0$ be the nodes with those quantiles in the beginning. Note that $L_i$ and $H_i$ are the nodes whose values are not our targets. We will show the quantities of $\frac{|L_i|}{n}$ and $\frac{|H_i|}{n}$ decrease in each iteration as our \texttt{3-TOURNAMENT} algorithm (Algorithm \ref{alg:median}) makes progress.

Initially, $l_0 = h_0 = \frac{1}{2} - \epsilon$. Let $h_{i+1} = 3h^2_i - 2h^3_i$ and $l_{i+1} = 3{l}^2_i - 2l^3_i$ for $i \geq 0$, we will show that $\frac{|L_i|}{n}$ and $\frac{|H_i|}{n}$ concentrate around $l_i$ and $h_i$. Note that $h_i$ and $l_i$ roughly square in each iteration. Once they decrease below a constant after the first $O(\log (1/\epsilon))$ iterations, they decrease double exponentially in each iteration. The tournaments end when $l_i$ and $h_i$ decrease below $T = 1/n^{1/3}$.

\begin{algorithm}[H]
\begin{algorithmic}[1]
\STATE $h_0, l_0 \leftarrow \frac{1}{2}-\epsilon$ \\
\STATE $i \leftarrow 0$, $T = 1/n^{1/3}$.
\WHILE{$l_i > T$}
	\STATE $h_{i+1} \leftarrow 3h^2_i - 2h^3_i$, $l_{i+1} \leftarrow 3l^2_i - 2l_i^3$ \\
		\STATE Select three nodes $t_1(v)$, $t_2(v)$, $t_3(v)$ randomly. \\
		\STATE $x_v \leftarrow median(x_{t_1(v)}, x_{t_2(v)}, x_{t_3(v)})$.
\ENDWHILE
\STATE Sample $K=O(1)$ nodes uniformly at random and output the median value of these nodes.
\end{algorithmic} 
\caption{\texttt{3-TOURNAMENT}($v$)}
\label{alg:median}
\end{algorithm}

\begin{replemma}{3tournamentround}Let $t$ denote the number of iterations in Algorithm \ref{alg:median}. We have $t \leq  \log_{11/8} (\frac{1}{4\epsilon}) + \log_{2} \log_{4} n = O(\log(1/\epsilon) + \log\log n)$. \end{replemma}
\begin{proof}
Suppose that $(\frac{11}{8})^{i-1}\epsilon \leq 1/4$, we will show that $l_i \leq 1 - (\frac{11}{8})^{i}\epsilon$. First, $l_0 = 1 - \epsilon$. Suppose by induction that $l_{i-1} \leq 1 - (\frac{11}{8})^{i-1}\epsilon$. We have:
\begin{align*}
l_i &= 3 l^{2}_{i-1} - 2 l^{3}_{i-1} \\
&= 3 \left(\frac{1}{2} - \left(\frac{11}{8}\right)^{i-1} \cdot \epsilon \right)^2 - 2 \left(\frac{1}{2} - \left(\frac{11}{8}\right)^{i-1} \cdot \epsilon \right)^3 \\
&= \frac{1}{2} - \frac{3}{2} \left(\frac{11}{8}\right)^{i-1}\cdot \epsilon + 2 \left( \left(\frac{11}{8}\right)^{i-1} \cdot \epsilon\right)^3 \\
&\leq \frac{1}{2} - \frac{3}{2} \left(\frac{11}{8}\right)^{i-1}\cdot \epsilon + \frac{1}{8} \left( \left(\frac{11}{8} \cdot \epsilon \right)^{i-1}\right) \\
&= \frac{1}{2} - \left( \frac{11}{8} \right)^{i} \cdot \epsilon 
\end{align*}
Therefore, after $i_0 =  \log_{11/8} (\frac{1}{4\epsilon}) $ iterations, we have $l_i \leq \frac{1}{2} - \left(\frac{11}{8} \right)^{i_0} \epsilon \leq \frac{1}{2} - \frac{1}{4} = \frac{1}{4}$.

Suppose that $i_1 = \log_{2} \log_{4} n$, we have 
\begin{align*}
l_{i_0+ i_1} &\leq 3 l_{i_0 +i_1 -1}^2 \\
&\leq 3^{i_1} l_{i_0}^{2^{i_1}} \\
&\leq  3^{i_1} \cdot \left(\frac{1}{4} \right)^{2^{i_1}}\\ 
&\leq (\log_4 n)^{\log_2 3} \cdot \left(\frac{1}{4} \right)^{\log_4 n}\\
&= \frac{(\log_4 n)^{\log_2 3}}{n} \leq n^{2/3}
\end{align*}

Therefore, $t \leq i_0 + i_1 =  \log_{11/8} (\frac{1}{4\epsilon}) + \log_{2} \log_{4} n = O(\log (1/\epsilon) + \log\log n)$.
\end{proof}

\begin{lemma}For each iteration $0 \leq i < t$, $\E[\frac{|L_{i+1}|}{n} \mid L_{i}] = 3(\frac{|L_i|}{n})^2 -2(\frac{|L_i|}{n})^3$ and $\E[\frac{|H_{i+1}|}{n} \mid H_{i}] =3(\frac{|H_{i}|}{n})^2 -2(\frac{|H_{i}|}{n})^3$.\end{lemma}
\begin{proof}
We will show the proof for $\E[\frac{|L_{i+1}|}{n} | L_i]$, since $\E[\frac{|H_{i+1}|}{n} \mid H_{i}]$ is the same. Note that a node $v$ is in $L_{i+1}$ if and only if at least 2 ot $t_1(v)$, $t_2(v)$, $t_3(v)$ are in $L_{i}$. Therefore,
\begin{align*}
\E\left[\frac{|L_{i+1}|}{n} \mid L_{i}\right] &= \frac{1}{n} \sum_{v \in V} \left( \left(\frac{|L_{i}|}{n}\right)^3  + 3\left(\frac{|L_{i}|}{n}\right)^2 \left(1 - \frac{|L_{i}|}{n} \right) \right) 
=3\left(\frac{|L_i|}{n}\right)^2 -2\left(\frac{|L_i|}{n}\right)^3 \qedhere
\end{align*}
\end{proof}

The following lemma shows the probabilities that $\frac{|L_{i}|}{n}$ and $\frac{|H_{i}|}{n}$ deviate from their expectation by a $(1 + \epsilon')$ factor are polynomially small in $n$, provided $\epsilon'$ is sufficiently large. We cap the quantity at $T$ for the purpose of applying concentration inequalities.

\begin{replemma}{LHconcentrate}\label{lem:LH_concentrate}Let $\epsilon' = \Omega(\frac{(\log n)^{1/2}} { n^{1/3}})$, w.h.p.~for each iteration $0 \leq i <= t$, $\frac{|H_i|}{n} \leq (1+\epsilon') \cdot \max(T, \E[\frac{|H_i|}{n} \mid H_{i-1}])$ and $\frac{|L_i|}{n} \leq (1+\epsilon') \cdot \max(T, \E[\frac{|L_i|}{n} \mid L_{i-1}])$. \end{replemma}

\begin{proof}
In each iteration, since each node set its value independently, by Chernoff Bound (see Lemma \ref{lem:largedev}), we have 
\begin{align*}
	\Pr\left(\frac{|L_i|}{n} \leq (1 + \epsilon') \cdot \max(T, \frac{\E[|L_i| \mid L_{i-1}]}{n}) \right) &\leq \exp(-\Omega(\epsilon'^2 \max(n\cdot T, \E[|L_i| \mid L_{i-1}]) ) ) \\
	= \exp(-\Omega(\epsilon'^2 n^{2/3} ) )= 1/\poly(n)  
\end{align*}
The proof for $|H_i|$ is the same.
\end{proof}

Then we bound the cumulative deviation of $\frac{|H_i|}{n}$ from $h_i$ and the cumulative deviation of $\frac{|L_i|}{n}$ from $l_i$.
\begin{replemma}{error2}\label{lem:error2}Let $\epsilon' = \Omega(\frac{\log n^{1/2}} { n^{1/3}})$. W.h.p. for each iteration $0 \leq i <= t$, $\frac{|L_{i}|}{n} \leq \max((1+\epsilon')^{\frac{3^{i}-1}{2}} \cdot l_{i}, (1+\epsilon')T )$ and $\frac{|H_{i}|}{n} \leq \max((1+\epsilon')^{\frac{3^{i}-1}{2}}  \cdot h_{i+1}, (1+\epsilon')T )$.\end{replemma}
\begin{proof}
We only show the proof for $|L_i|$ and we will prove by induction. Initially, $\frac{|L_0|}{n} \leq l_0$. Suppose that $\frac{|L_{i}|}{n} \leq \max((1+\epsilon')^{\frac{3^{i}-1}{2}} \cdot l_{i+1}, (1+\epsilon')T )$ is true.

By Lemma \ref{lem:LH_concentrate}, we have w.h.p., $\frac{|L_{i+1}|}{n} \leq (1+ \epsilon') \cdot \max(\frac{\E[|L_{i+1}| \mid L_{i}]}{n}, T)$. If $\frac{\E[|L_{i+1}| \mid L_{i}]}{n} \leq T$, then we are done, since $\frac{|L_{i+1}|}{n} \leq (1+ \epsilon') \cdot \frac{\E[|L_{i+1}| \mid L_i]}{n} = (1+\epsilon')\cdot T$.  Also, if $\frac{|L_{i}|}{n} \geq (1+\epsilon')T$, then $\frac{\E[|L_{i+1}| \mid L_{i}]}{n} \leq 3 \cdot \frac{|L_i|^2}{n^2} \leq  \frac{|L_i|}{n} \cdot O(\frac{1}{n^{2/3}}) \leq T$. 

Otherwise, if $\frac{|L_{i}|}{n} > (1+\epsilon')T$ then it must be the case that $\frac{|L_{i}|}{n} \leq (1+\epsilon')^{\frac{3^{i}-1}{2}} \cdot l_{i+1}$ by induction hypothesis. Therefore,
\begin{align*}
	\frac{|L_{i+1}|}{n} &\leq (1+\epsilon') \cdot \frac{\E[|L_i| \mid L_{i-1}]}{n} \\
	&\leq (1+\epsilon') \cdot 3\left(\frac{|L_i|}{n}\right)^2 -2\left(\frac{|L_i|}{n}\right)^3 \\
	&\leq (1+\epsilon') \cdot \left( 3\left((1+\epsilon')^{\frac{3^{i}-1}{2}}\cdot l_i \right)^2 -2\left((1+\epsilon')^{\frac{3^{i}-1}{2}}\cdot l_{i} \right)^3 \right) \\ 
	&\leq (1+\epsilon') \cdot \left( \left((1+\epsilon')^{\frac{3^{i+1}-3}{2}} \right) \cdot \left( 3\cdot  l_i^2  -2 \cdot l^3_{i} \right) \right) =  \left((1+\epsilon')^{\frac{3^{i+1}-1}{2}} \right) \cdot l_{i+1} \qedhere  
\end{align*}
\end{proof}

Now, by setting $\epsilon'$ roughly equal to $1/3^{t} \sim \epsilon^{3.45} / (\log_4 {n})^{1.59}$, we can bound the final deviations of $\frac{|H_i|}{n}$ and $\frac{|L_i|}{n}$ from $T$ by a factor of 2 w.h.p. 
\begin{replemma}{nobad}\label{lemma:nobad} Let $\epsilon = \Omega(\frac{\log^{0.61} n}{n^{0.096}})$ and $\epsilon' = \epsilon^{3.45} / (\log_4 {n})^{1.59}$, then $\frac{|L_{t}|}{n} \leq 2T$ and $\frac{|H_{t}|}{n} \leq 2T$ w.h.p.  \end{replemma}
\begin{proof}

First, $\epsilon' = \Omega(\frac{\epsilon^{3.45}}{\log^{1.59} n}) = \Omega(\frac{\log^{1/2} n}{ n^{1/3}}) $.
By Lemma \ref{lem:error2}, w.h.p.~we have either $\frac{|L_t|}{n} \leq (1+\epsilon')\cdot T$ or $\frac{|L_t|}{n} \leq (1+\epsilon')^{\frac{3^{i}-1}{2}} \cdot l_{t}$. If it is the former, then we are done since $\frac{|L_t|}{n} \leq (1+\epsilon')\cdot T \leq 2T$. Otherwise, we have
\begin{align*}
	\frac{|L_{t}|}{n} &\leq \left((1+\epsilon')^{\frac{3^{t}-1}{2}} \right) \cdot T \\
	&\leq \left( 1 + 2\cdot \epsilon' \cdot \frac{3^{t}-1}{2} \right)\cdot T && \mbox{since $(1+x)^n \leq 1+2nx$ for $nx \leq \frac{1}{2}$} \\
	&\leq \left( 1 +  \epsilon' \cdot {3^{ \log_{11/8} (\frac{1}{4\epsilon}) + \log_{2} \log_{4} n }} \right)\cdot T \\
	&\leq \left(1+ \epsilon' \cdot \left(\frac{1}{4\epsilon}\right)^{\log_{11/8} 3} \cdot \log_4^{\log_2 3} n \right) \cdot T \\
	&\leq 2 \cdot T && \epsilon' = \epsilon^{3.45} / (\log_4 {n})^{1.59} \qedhere
\end{align*}
\end{proof}

Finally, we show that when $\frac{|L_t|}{n}$ and $\frac{|H_t|}{n}$ are $O(1/n^{2/3})$, if we sample a constant number of values randomly and output the median of them, then w.h.p.~the median is in $M_t$.

\begin{lemma}\label{lem:last2}W.h.p.~every node outputs a quantile in $[\frac{1}{2} - \epsilon, \frac{1}{2} + \epsilon]$.\end{lemma}
\begin{proof}
By Corollary \ref{lemma:nobad}, w.h.p.~at the end of iteration $t$, $\frac{|L_t|}{n} \leq 2/n^{2/3}$ and $\frac{|H_t|}{n} \leq 2/n^{2/3}$. The algorithm outputs a quantile in $[\frac{1}{2} - \epsilon, \frac{1}{2} + \epsilon]$ if there are less than $K/2$ nodes in $L_t$ are sampled and less than $K/2$ nodes in $H_t$ are sampled.

The probability that at least $K/2$ nodes in $|L_t|$ are sampled is at most
\begin{align*}
\binom{K}{K/2} \cdot \left(\frac{2}{n^{2/3}} \right)^{K/2} &\leq \left(\frac{eK}{\frac{K}{2}}\right)^{K/2}\cdot \left(\frac{2}{n^{2/3}} \right)^{K/2} 
\leq \left(\frac{4e}{n^{2/3}} \right)^{K/2} 
\end{align*}
Similarly, the probability that at least $K/2$ nodes in $|H_t|$ are sampled is also at most $\left(\frac{4e}{n^{2/3}} \right)^{K/2}$. By an union bound, the probability that less than $K/2$ nodes in $L_t$ are sampled and less than $K/2$ nodes in $H_t$ are sample is at least $1 - 2\cdot \left(\frac{4e}{n^{2/3}} \right)^{K/2} = 1- 1/\poly(n)$.
\end{proof}

Theorem \ref{thm:quantile_approx} follows from Lemma \ref{lem:last} and Lemma \ref{lem:last2}.

\section{Exact Quantile Computation}
To fill the gap of approximating the $\phi$-quantile with $\epsilon$ error for $\epsilon = O(1/n^{0.096})$, we show that the exact quantile computation can be done in $O(\log n)$ rounds using $O(\log n)$ message size. Since we are to compute the exact quantile $\phi$, we can assume that $k_0 \defeq \phi \cdot n$ is an integer. The problem is to compute a value whose rank is $k_0$. Again, w.l.o.g.~we assume that every value has a distinct value initially.
\begin{algorithm}[H]
\begin{algorithmic}[1]
\STATE $k_0 \leftarrow \phi \cdot n$
\FOR{$i = 1,2, \ldots ,25$}
	\STATE \label{line:1} Each node $v$ computes an $\frac{\epsilon}{2}$-approximate of the $(\frac{k_{i-1}}{n}-\frac{\epsilon}{2})$-quantile and an $\frac{\epsilon}{2}$-approximate of the $(\frac{k_{i-1}}{n}+\frac{\epsilon}{2})$-quantile with $\epsilon = n^{-0.05}/2$.
	\STATE \label{line:2} Each node learns the max and the min of these approximates of all nodes.
	\STATE \label{line:2.5} Compute the rank of $\min$ among the original $x_v$'s and denote it by $R$. 
	\STATE \label{line:4} Each node $v$ set $x_v \leftarrow \infty$ if $x_v \notin [\min,\max]$. Call these nodes {\it valueless}, otherwise {\it valued}. 
	\STATE  \label{line:5} Let $m_i$ be the smallest power of 2 that is larger than $(n^{0.99}/2)/(\mbox{\# valued nodes})$. Each valued node makes $m_i$ copies of its value and distribute them to valueless nodes so that there are at least $n^{0.99} / 2$ valued nodes. (For convenience, we let the duplicated values to have smaller ranks than the original one.)
	\STATE \label{line:6} Set $k_i \leftarrow m_i\cdot (k_{i-1} - R + 1)$.
\ENDFOR
\STATE Every node outputs an $(\epsilon/3)$-approximate $(\frac{k_{25}}{n} - \frac{\epsilon}{2})$-quantile.
\end{algorithmic} 
\caption{Exact Quantile Computation}
\label{alg:exact}
\end{algorithm}

The following is a detailed implementation of each step. 

	 {\bf Step \ref{line:1}}: Since $\epsilon = n^{-0.05}/2$, we can $\frac{\epsilon}{2}$-approximate the quantiles in $O(\log n)$ rounds by Theorem \ref{thm:quantile_approx}.

	 {\bf Step \ref{line:2}}: The maximum (the minimum) can be computed by having each node forwarding the maximum (the minimum) value it has ever received. Since it takes $O(\log n)$ rounds to spread a message by \cite{FG85, Pittel87}, this step can be done in $O(\log n)$ rounds.
	
		 {\bf Step \ref{line:2.5}}: The rank of the minimum can be computed by performing a counting. The nodes whose $x_v$ values are less than or equal to the minimum will be assigned 1. Otherwise they are assigned 0. By Kempe et al.~\cite{KDG03}, the sum can be aggregated in $O(\log n)$ rounds w.h.p.
	
	 {\bf Step \ref{line:5}}: Consider the following process for distributing the values. Initially, every valued node $v$ generates a token with weight equal to $m_i$, the targeted number of copies. We denote the token by a value-weight pair $(x_v, m_i)$. Recall that $m_i$ is a power of 2 and it can be computed by counting the number of valued nodes in $O(\log n)$ rounds. The goal is to split and distribute the tokens so that every node has at most one token of weight 1 in the end.
	
	The process consists of $O(\log n)$ phases and each phase uses $O(1)$ rounds. Suppose that node $v$ holds tokens $(x_1,w_1), (x_2,w_2), \ldots (x_m, w_m)$ at the beginning of a phase. For each token $(x_i, w_i)$, if $w_i \neq 1$, $v$  splits $(x_i, w_i)$ into two tokens $(x_i, w_i/2)$ and push each one to a node randomly. If $w_i = 1$, then the token remains at $v$. Note that when two tokens of the same value are on the same node, they do not merge. 
	
	First note $\lg m_i = O(\log n)$ phases are needed to split to tokens into tokens of weight 1. Now, we show that it takes constant number of rounds to implement one phase. 

	Let $N(v,i)$ denote the number of tokens at $v$ at the end $i$'th phase. Let $N = \max_{v,i} N(v,i)$. It takes $N$ rounds to implement one phase. The value of $N(v,i)$ can be bounded by following calculation. Since at any phase, there cannot be more than $n^{0.99}$ tokens and each token appears at a node chosen uniformly at random, we have
\begin{align*}
	\Pr(N(v,i) \geq 100K) &\leq \binom{n^{0.99}}{100K} \cdot \frac{1}{n^{100K}} \leq \left(\frac{n^{0.99}}{n} \right)^{100K} = (1/n)^{K}.
\end{align*}
	By taking an union bound over each $v$ and $i$, we conclude that $N < 100K$ w.h.p. Therefore, w.h.p.~each phase can be implemented in $200K = O(1)$ rounds.	
	At the end of $O(\log n)$ phase, every node has at most $100K$ tokens of weight 1 w.h.p.

	Next, in the subsequent phases, if a node holds more than one token, then it will push every token except one to a random node (one random node per token). 	We say a token succeeded at the end of the phase if it was pushed to a node without any other tokens.  Consider a token, the probability a token did not succeed in a phase is at most $n^{0.99} / n = 1/n^{0.01}$, since there are at most $n^{0.99}$ nodes with tokens at the end of the phase.
	
Therefore, after $100K$ phases, the probability that a token did not succeed is at most $1/n^{0.01\cdot 100K} = 1/n^{K}$. By a union bound over the tokens, we conclude that w.h.p.~all tokens succeeded. Moreover, again, each phase can be implemented using $100K$ rounds. Therefore, we can split and distribute the tokens in $O(\log n)$ rounds. Each original token is duplicated with exactly $m_i$ copies.  Then, those nodes holding a token set its value to the value of the token. Nodes without a token will remain valueless.

	\paragraph{Correctness: }
Let $\texttt{ans}$ be the answer (the value whose initial rank is $k_0$). Let $M_i = \prod_{j=1}^{i} m_i$ denote the number of copies of each value from the beginning. We show by induction that the values whose ranks lying in $(k_i - M_i ,k_i]$ are $\texttt{ans}$ after iteration $i$. Initially, $M_0 = 1$, the statement trivially holds. 

Suppose that at the end of iteration $i-1$, the values whose ranks lying in $(k_{i-1} - M_i ,k_{i-1}]$ are $\texttt{ans}$.  After Step \ref{line:4} in iteration $i$, the rank of $\texttt{ans}$ is exactly $k_{i-1}-R+1$.  Since every value is duplicated to have $m_i$ copies after Step \ref{line:5}, the rank of $\texttt{ans}$ becomes $k_{i} = m_i \cdot (k_{i-1}-R+1)$ (Recall that we let the original value to have a larger rank than the duplicated values).  Since the values whose ranks lying in $(k_{i-1} - M_{i-1} ,k_{i-1}]$ were $\texttt{ans}$ at the end of iteration $i-1$ and $\texttt{ans} \in [\min, \max]$, it must be the case that every value in this range is duplicated with $m_i$ copies. Therefore, all values whose ranks lying in $(k_{i} - m_i \cdot M_{i-1}, k_i] = (k_{i} - M_i, k_i]$ are $\texttt{ans}$ after iteration $i$.
	
	

Next, we show that after 25 iterations, $M_{25} \geq \epsilon n$. Suppose that $M_{i-1} \leq \epsilon n$. The number of valueless node after Step \ref{line:4} is at most $2\epsilon n + 2M_{i-1} \leq  4\epsilon n$. The $2\epsilon n$ comes from the fact that at most $2\epsilon n$ values can lie in the range $(\min, \max)$. The $2M_{i-1}$ term is because at most $M_{i-1}$ values are equal to min and max. Therefore, $m_i \geq (n^{0.99}/2)/(4\epsilon  n) = (n^{0.04}/4)$ and $M_i \geq (n^{0.04}/4)\cdot M_{i-1}$.
	
Therefore, $M_{25} \geq \min(\epsilon n, (n^{0.04}/4)^{25}) = \epsilon n$. Now the the values of the items whose ranks are in the range of $(k_{25} - \epsilon n, k_{25}]$ must be \texttt{ans}. We then compute a quantile in $(\frac{k_{25}}{n} - \epsilon, \frac{k_{25}}{n}]$ using our $\epsilon/3$-approximate quantile computation algorithm to approximate the $(\frac{k_{25}}{n} - \frac{\epsilon}{2})$-quantile in $O(\log n)$ rounds.

\section{A Lower Bound}
We show that it takes at least $\Omega(\log(1/\epsilon) + \log \log n)$ rounds to approximate the $\phi$-quantile up to $\epsilon$ error w.h.p.  

\repeattheorem{lowerbound}
\begin{proof}
Consider the following two scenarios. The first is when each node is associated with a distinct value from $\{1,2,\ldots, n\}$. The second is when each node is associated with a distinct value from $\{1 + \lfloor 2 \epsilon n \rfloor, \ldots, n + \lfloor 2 \epsilon n \rfloor \}$.

A node is able to distinguish between these cases only if it recieves an value from $S \defeq \{1,2, \ldots,  1 + \lfloor 2 \epsilon n \rfloor \} \cup \{n+1,n+2, \ldots,  n+ \lfloor 2 \epsilon n \rfloor \}$. Otherwise, it can only output a value whose quantile is in $[1/2 - \epsilon, 1/2 + \epsilon]$ with probability 1/2, since the difference of the $\phi$-quantile in both scenarios is at least $\lfloor 2\epsilon n \rfloor \geq \epsilon n$. 

Call a node {\it good} if it ever received a value from $S$, and {\it bad} otherwise. Note that a bad node cannot outputs a correct answer with probability more than $1/2$. Initially there are at most $2\cdot \lfloor 2 \epsilon n \rfloor \leq 4\epsilon n$ good nodes. We will show that with probability $1-1/n$ there exists at least one bad node at the end of round $t$.

Let $X_0 = 2\cdot \lfloor 2 \epsilon n \rfloor$ and let $X_i$ denote the number of good nodes at the end of round $i$. Given a bad node $v$, it can become good if it pulls from a good node or some good node pushes to it. Let $Y_v$ denote the event that $v$ pulls from a good node, we have $\Pr(Y_v \mid X_i) = X_i /n$. Also, the pushes from the good nodes can only generate another at most $X_i$ good nodes. Therefore,
$$\E[X_{i+1} \mid X_i ] \leq 2X_i + \E\left[\sum_{v \in B_i} Y_v\right] \leq 3X_i, \qquad \mbox{since $\E[\sum_{v \in B_i} Y_v \mid X_i] \leq X_i$}$$

By Chernoff Bound, we have $\Pr(\sum_{v \in B_i} Y_v > 2X_i \mid X_i) \leq e^{-X_i/2} \leq e^{-5\log n} \leq 1/n^5$, since $X_i \geq X_0 \geq 10\log n$. Therefore, with probability at least $1 - 1/n^5$, $X_{i+1} \leq 4X_i$. By taking a union bound over such events for the first $t' = \log_{4} (8/\epsilon)$ rounds, we conclude with probability at least $1-1/n^4$, $X_{t'}\leq (4\epsilon n) 4^{t'} \leq n/2$.

Let $t_0$ be the last round such that $X_{t_0} \leq n/2$. Define $Z_i = |B_i| / n$. A node $v$ remains in $B_i$ if it did not pull from a good node {\bf and} it was not pushed from any good nodes. Denote the event by $W_v$, we have $\Pr(W_v \mid B_i) \geq Z_i \cdot (1 - \frac{1}{n})^{n-1} \geq Z_i\cdot e^{-1}$. Therefore, $$ \E[|B_{i+1}| \mid B_i] = \E\left[ \sum_{v \in B_i} W_v \mid B_i \right]   \geq  \left(\sum_{v \in B_i} Z_{i} \cdot e^{-1} \right)  = |B_i| \cdot Z_{i} \cdot e^{-1}$$

Note that the events $\{W_v\}_{v \in B_i}$ are negatively dependent \cite{DR98} (the number of empty bins in the balls into bins problem). Suppose that $Z_i \cdot |B_i| \geq 60e \log n$. By Chernoff Bound,
$$\Pr\left(|B_{i+1}| \leq |B_i|\cdot Z_i / (2e) \mid B_i\right) \leq e^{-Z_i \cdot |B_i| /(12e)} \leq 1 / n^5$$
Therefore, $\Pr(Z_{i+1} \leq Z^2_i / (2e) \mid Z_i) \leq 1/n^5$. Suppose that $t_1 =O(1)$. We can take an union bound over the subsequent $t_1$ rounds to show that with probability at least $1 - 1/n^4$, $Z_{i+1} \geq Z^2_i/(2e)$ holds for these rounds, as long as $Z^2_{t_0 + t_1} \geq (60\log n) / n$.

Let $z_{t_0} = 1/2$ and $z_{i+1} = z^2_{i} / (2e)$. If $z^2_{t_0 + t_1} \geq (60\log n) / n$, with probability at least $1-1/n^4$, we have $Z_{t_0+t_1} \geq z_{t_0+t_1}$. Let $t_1 = \frac{1}{2}\lg \log n$, by definition of $z_i$, we have \begin{align*} z_{t_0 + t_1} &= \left(\frac{1}{2e}\right)^{t_1} \cdot z_{t_0}^{2^{t_1}} \\
&\geq \left(\frac{1}{2e}\right)^{t_1+ 2^{t_1}}  && z_{t_0} = 1/2 \geq 1/(2e) \\
&= \left(\frac{1}{2e}\right)^{(\lg \log n)/2 + \sqrt{\log n}} \\
&\geq \left(\frac{1}{2e}\right)^{2\sqrt{\log n}} && \mbox{for sufficiently large $n$}
\end{align*}

Since $(2e)^{-2\sqrt{\log n}} = \Omega((\log n) /n)$, $z_{t_0+t_1} \geq 2^{-2\sqrt{\log n}} \geq (60\log n)/n$ for sufficiently large $n$. Therefore, with probability at least $1 - 1/n^4$, $Z_{t_0+t_1} > 0$.

Since with probability $1 - 1/n^4$, $t_0 \geq \log_4{8/\epsilon}$. By taking an union bound over these two events, we conclude with probability at most $2/n^4$, all nodes are good by the end of round $t_0 + t_1$.

Therefore, the probability that every node computes a correct output at the end of round $t_0 + t_1$ is at most $1/2 + 2/n^4$.
\end{proof}

\section{Robustness}
In this section, we show that our algorithm is robust against failures. We consider the following model of failures: Let $0< \mu < 1$ be a constant. 
Every node $v$ in every round $i$ is associated with a probability $0 \leq p_{v,i} \leq \mu$. Note that $p_{v,i}$ is pre-determined before the execution of the algorithm. During the execution of round $i$, each node $v$ fails with probability $p_{v,i}$ to perform its operation (which may be either push or pull). 

\repeattheorem{robustness}

We show how to modify the algorithm and the analysis for the tournament algorithms and the exact quantile computation algorithms in the following.
\subsection{The Tournament Algorithms}\label{sec:robustness1}



First consider the \texttt{2-TOURNAMENT} algorithm. Initially, every node is {\it good}. Consider node $v$, instead of only pulling from 2 neighbors each iteration, now pull from $\Theta(\frac{1}{1-\mu} \cdot \log(\frac{1}{1-\mu}))$ neighbors. We say a pull is {\it good} if the node performing the pull operation did not fail {\bf and} it pulls from a node who is good at the end of iteration $i-1$. A node remains good at the end of iteration $i$ if there are at least two good pulls. Then, $v$ uses the first two good pulls to execute the tournament procedure. Also, note that in the last iteration, we let a good node to have probability of $\delta$ to do the two tournament using the first two good pulls and probability $1-\delta$ to set the value equal to the first good pull. We show that the good nodes always consists of a constant fraction of nodes. 

\begin{lemma} For $0 \leq i \leq t-1$, w.h.p.~at the end of each iteration $i$ of {\normalfont \texttt{2-TOURNAMENT}} , there are at least $n/2$ good nodes if every node pulls from $k = \Theta(\frac{1}{1-\mu} \cdot \log(\frac{1}{1-\mu}))$ other nodes. \end{lemma}
\begin{proof}
We will prove by induction. Initially, every node is good. Suppose that there are at least $n/2$ good nodes at the end of iteration $i-1$. During iteration $i$, a node with more than $k-2$ bad pulls will become bad at the end of iteration $i$.

Let $k = \frac{4}{1-\mu} \log \frac{4}{1-\mu} + 1$. The probability that there are at least $k-1$ bad pulls is at most:
\begin{align*}\binom{k}{k-1} \cdot \left(1 - \frac{(1-\mu)}{2} \right)^{k-1} &\leq k \cdot e^{-\frac{(1-\mu)}{2} \cdot (k-1)}\\
& \leq \left(\frac{4}{1-\mu} \log \frac{4}{1-\mu} + 1\right) \cdot \left(\frac{1-\mu}{4}\right)^2\\
& \leq 1/e + 1/16 \leq 0.44 && \mbox{$x \log (1/x)$ maximized at $1/e$} \end{align*}
Therefore, the expected number of bad nodes is at most $0.44 n$ for iteration $1\leq i \leq t-1$. Since each node becomes bad independently of other nodes, by Chernoff Bound, w.h.p.~there are at most $0.5n$ bad nodes at the end of iteration $i$. 
\end{proof}

At the end of iteration $t-1$, there are at least $n/2$ good nodes w.h.p. The expected number of bad nodes at the end of iteration $t$ is at most
$\sum_{v \in V} \delta \cdot 0.44 + (1-\delta)/2 \leq 0.5n$. By Chernoff Bound again, we can conclude that w.h.p.~there are at least $n/3$ good nodes. 

We can also modify the process in the same way for the \texttt{3-TOURNAMENT} algorithm. That is, in each iteration each node  pulls from $\Theta(\frac{1}{1-\mu} \log \frac{1}{1-\mu})$ other nodes. If there are less than 3 good pulls, then the node becomes bad. Otherwise, it uses the first 3 good pulls to do the tournament procedure. We can show similarly that there are at most a constant fraction bad nodes in each iteration. 

Consider a node $v$ in the modified processes. Suppose that $v$ is good in iteration $i$, then $v$ must have at least two (or three) good pulls. Note that the probability that $v$ pulls from a particular good node, conditioned $v$ is good at the end of iteration, is uniform among all nodes that are good at the end of iteration $i-1$.

Let $V_i$ denote the set of good nodes and $n_i = |V_i|$. Given any subset of good nodes $S \subseteq V_{i-1}$, the probability of choosing a node in $|S|$ is therefore $|S| / n_i$. Therefore, we can replace $\frac{|L_i|}{n}$, $\frac{|H_i|}{n}$, and $\frac{|M_i|}{n}$ in our proofs in the previous section with $\frac{|L_i|}{n_i}$, $\frac{|H_i|}{n_i}$, and $\frac{|M_i|}{n_i}$ and observe that all the statements hold. Note that $n_i \geq n/3 = \Omega(n)$ for $0\leq i \leq t$ w.h.p. Thus, all the concentration inequalities also hold.

In the last step of Algorithm \ref{alg:median}, all nodes pulls from $\Theta(\frac{K}{1-\mu} \log \frac{K}{1-\mu})$ nodes. If there are $K$ good pulls, then each node outputs the median directly. Otherwise, it becomes bad and outputs nothing. We can similarly show that there are at least constant fraction of good nodes. Therefore, at least a constant fraction of nodes output a correct answer and all the others output nothing. Note that we can use additional $O(t)$ rounds of pulling from valueless nodes to have al}l but $\frac{n}{2^{t}}$ nodes learn a correct answer.

\subsection{Exact Quantile Computation}\label{sec:robustness2}

Consider Algorithm \ref{alg:exact}. We know Step \ref{line:1} (Section \ref{sec:robustness1}),  Step \ref{line:2} \cite{ES09}, and Step \ref{line:2.5} \cite{KDG03} tolerate such failures with a constant factor delay on the running time.

The only step that remains to discuss is Step \ref{line:5}. We run the same algorithm described. Initially each valued node $v$ generates a token $(x_v,m_i)$. In the first $O(\log n)$ phases, each node tries to split every token whose weight is larger than 1 into two tokens with halved weight, and push one of them to a random node. If the push succeed, the token has been pushed is considered a {\it new token} and the one with halved weight remaining at the same node is considered as the {\it old one}. If the push fails, then we will merge them back to the original one. 


We will show that at the end of $O(\log n)$ phase, each node holds at most $O(1)$ tokens and their weights are 1. In each subsequent phases, if a node has more than one tokens, then it will push every token except one to a random neighbor. We argue that after $O(\log n)$ phases, every node contains exactly one token. 

 First, we show that the number of tokens at each node is bounded by a constant w.h.p.~so that each phase can be implemented in $O(1)$ rounds. Recall that $N(v,i)$ is the number of tokens at $v$ at the end of $i$'th phase. Suppose the total number of phases is $C \log n$ for some constant $C > 0$. For a token (except the one that is initially at $v$) to be at node $i$, it must be the case that it was pushed to $i$ during one of the phase. Note that at any phase, there are at most $ n^{0.99}$ tokens. By taking an union bound over all possible set of tokens of size $200K$ and all possible combinations of phases on when each of the token was pushed to $i$, we have 
 $$	\Pr(N(v,i) \geq 200K+1) \leq \binom{n^{0.99}}{200K}\cdot (C \log n)^{200K} \cdot \frac{1}{n^{200K}} \leq \left(\frac{C n^{0.99} \log n}{n} \right)^{200K} \leq 1/n^{K}. $$
 
 Next, we show that for a token $(x_v,m_i)$, it takes $O(\log n)$ rounds to split into tokens of weight 1 w.h.p. Let $T_v(i)$ denote the set of all $v$'s tokens with weight at least 2 at the end of Phase $i$. Let $\Phi(i) = \sum_{(x_v,w) \in T_v(i)} w^2$ be a potential function.  Consider a token $(x_v,w) \in T_v(i)$, since with probability at most $\mu$ it fails to split, the expected contribution of the token to the $\Phi(i+1)$ is at most:
 $$\mu w^2 + (1-\mu) 2\cdot (w/2)^2 =  \left(1 - \frac{1-\mu}{2}\right) \cdot w^2 $$
   Therefore, $\E[\Phi(i+1) \mid \Phi(i)] \leq (1-(1-\mu)/2) \cdot \Phi(i)$. Since $\Phi(0) \leq n$, after $t= \frac{2K }{1-\mu} \log n$ rounds, we have 
$$\E[\Phi(t)] \leq \left(1 - \frac{1-\mu}{2}\right)^{\frac{2K }{1-\mu} \log n} \cdot n \leq e^{-K\log n} \cdot n \leq 1/n^{K-1} $$
	
	Since $\Phi(t)$ must be an integer, we conclude that with probability at least $1 - 1/n^{K-1}$, $\Phi(t) = 0$ by Markov's inequality. Therefore, after $t$ rounds w.h.p.~every token of value $v$ has weight 1. We can further take an union over all tokens to show that it holds for all tokens w.h.p.
	
	Therefore, w.h.p.~the weight of all tokens are 1 at the end of $O(\log n)$ round. In the subsequent phases, the probability a token fail to be pushed to a node without any other tokens is at most $\mu + n^{0.99}/n = O(1)$. Therefore, the probability a token fails in the all of the next $O(\log n)$ phases is at most $1/\poly(n)$. Thus, w.h.p.~after $O(\log n)$ phases, every node has at most one token. Since each phase can be implemented using $O(1)$ rounds, we conclude that Step \ref{line:5} can be done in $O(\log n)$ rounds in our failure model.
	
\section*{Acknowledgement} We thank Frederik Mallmann-Trenn for pointing out \cite{DGMSS11} to us.

\bibliographystyle{alpha}
\bibliography{references}

\clearpage

\appendix












\section{A Sampling Algorithm}\label{apx:sampling}
We first present a simpler algorithm with a higher message complexity of $O( \log^2 n   / \epsilon^2 )$ and $O(\log (1/\epsilon) + \poly(\log\log n))$ running time. In Manku et al.~\cite{MRL99}, they observed the following. Suppose that we sample a multi-set $S$ of nodes (with replacement) from all the node uniformly at random. Then, the $\phi$'th quantile of $S$ has quantile approximately equal to $\phi$ in all the nodes.  

\begin{lemma}\label{lem:sample} Suppose that a multi-set $S$ of values are sampled from the $n$ values and $|S| = \Omega(\log n / \epsilon^2)$. Given $z$ (which is not necessarily in $S$), define $Q_S(z)$ to be the number of elements in $S$ whose values are smaller than $z$. We have $|Q_S(z) - Q(z)| \leq \epsilon$ w.h.p. \end{lemma}


Therefore, we can $\epsilon$-approximate the $\phi$-quantile or $\epsilon$-approximate the quantile of every node in $O(\log n / \epsilon^2)$ rounds. Each node $v$ uniformly samples $|S_v|$ nodes in $|S_v|$ rounds. To approximate the $\phi$-quantile, output the element $z \in S$ such that $Q_S(z)$ is the closest to $\phi$. Since there are at least $1/\epsilon$ elements in $S$, we must have $|Q_S(z) - \phi| \leq \epsilon$. Therefore, $|\phi - Q(z)| \leq |\phi - Q_{S}(z)| + |Q_S(z) -Q(z)| \leq 2\epsilon$. To approximate $v$'s own quantile, simply output $Q_S(x_v)$. We can take an union bound over each node to show w.h.p.~every node outputs correctly. The algorithm uses $O(\log n)$ message size, since in each round each node is only pulling the value of the node contacted.





By using larger message sizes, it is possible to reduce the number of rounds. Consider the following {\bf doubling algorithm}: Each node $v$ maintains a multi-set buffer $S_v$.  Before the first round, each node $v$ sample a random node $t_0(v)$ and set $S_v(0) = \{t_0(v)\}$.
In round $i$, let $t_i(v)$ be the node contacted by $v$ in that round, node $v$ sets $S_v(i) \leftarrow S_v(i) \cup S_{t_i(v)}(i)$. Let $S_v$ denote $S_v(T)$ at the end of the algorithm.

Since the buffer size doubles each round after the first round, after $t = O(\log ((\log n) /\epsilon^2) ) = O(\log \log n+ \log (1/\epsilon))$ rounds, the size of the buffer is $\Omega(\log n /\epsilon^2)$. The message complexity is $O(\log^2 n)$ for this algorithm. On the surface, it seems that we can apply Lemma \ref{lem:sample} to show that $Q_{S_v}(z)$ is a good approximation of $Q(z)$, since $|S_v| = \Omega((\log n)/\epsilon^2 )$. However, the elements in $S_v$ are not independently uniformly sampled. There are correlations among the elements. In the rest of the section, we will show that the correlations among the elements are not too strong so that the algorithm still works. 



\begin{lemma}\label{lem:double_sample} Let $\epsilon = \Omega(1/n^{1/16})$. Suppose that $S_v$ is sampled by using the doubling algorithm described above for $T = O(\log \log n + \log (1/\epsilon^2))$.  Given $z$, w.h.p.~we have $|Q_{S_x}(z) - Q(z)  | \leq \epsilon$. \end{lemma}

\begin{proof}
For convenience, we also define the multiset $S'_v(0) \leftarrow \{v\}$ and $S_v'(i) \leftarrow S'_v(i) \cup S'_v(t_i(v))$.  Also let $S'_v = S'_v(T)$.
Therefore, $S_v = \{t_0(u) \mid u \in S'_v \}$. First we show that the multiplicity of every element in $S'_v$ is bounded by some constant w.h.p. Then, we will reveal $t_0(u)$ for all $u \in S'_v$ to show that $Q_{S_v}(z)$ is concentrated around $Q(z)$. 

Given a node $w \in V$, we will show that w.h.p.~for every $x \in V$, the multiplicity of $w$ in $S'_v$, $m_{S'_{v}(T)}(w)$, is bounded above by a constant. 
Let $M_i = \max_y m_{S'_y(i)}(w)$. Observe that $M_{i+1} \leq 2 \cdot M_{i}$ due to the update rule. Let $S^{-1}(w,i) =  \bigcup_{y : w \in S'_y(i)} \{y\}$. Note that $M_{i} > M_{i-1}$ only if a node in $S^{-1}(w,i-1)$ chooses a node in $S^{-1}(w,i-1)$. We will first bound the size of $S^{-1}(w,i)$ over the execution of the algorithm, then we can bound the probability that $M_i > M_{i-1}$ increases. 

Initially, $S^{-1}(w,0) = \{w\}$. In round $i$, if a node $z$ chose a node in $S^{-1}(w,i-1)$, then we must have $z \in S^{-1}(w,i)$. Let $X_z$ denote whether the node $z$ has chosen a node in $S^{-1}(w,i-1)$.
\begin{align*}
\E[|S^{-1}(w,i)|] &\leq |S^{-1}(w,i-1)| + \sum_{z \in V} \E[X_z] \\
&\leq |S^{-1}(w,i-1)| + n \cdot |S^{-1}(w,i-1)| / n = 2 \cdot |S^{-1}(w, i-1)|.
\end{align*}

Moreover, we show that $X_z \leq 4 \cdot |S^{-1}(w,i-1)| + O(\log n)$ w.h.p. Let $k = 3 \cdot |S^{-1}(w,i-1)| + K \cdot \log n$. We have
\begin{align*}
\Pr(\sum_{z \in V} X_z \geq k ) &\leq \binom{n}{k}\left(\frac{|S^{-1}(w,i-1)|}{n} \right)^{k} \\
&\leq \left( \frac{en}{k}  \cdot \frac{|S^{-1}(w,i-1)|}{n}  \right)^{k} \\ 
& \leq \left( \frac{e \cdot |S^{-1}(w,i-1)|}{3 \cdot |S^{-1}(w,i-1)|}  \right)^{K \cdot \log n} \\
& \leq 1/\poly(n) \qedhere
\end{align*}

Therefore, w.h.p. $|S^{-1}(w,i)| \leq 4 \cdot |S^{-1}(w,i-1)| + O(\log n) $ for $i = 1 \ldots T$. As a result, $S^{-1}(w,T) = O(4^{K} \cdot \log n )  = O(\log^3 n/ \epsilon^4)$.

Recall that $M_i > M_{i-1}$ only if some node in $S^{-1}(w,i-1)$ selects a node in $S^{-1}(w,i-1)$ during round $i$.
Thus, the probability that $M_i > M_{i-1}$ is at most 
\begin{align*}
p \defeq 1 - (1 - O(\log^3 n / \epsilon^4) / n)^{ O(\log^3 n/ \epsilon^4) }  = O\left(\frac{\log^6 n}{\epsilon^8 n}  \right)
\end{align*}
The probability that $M$ increases in more than $K$ rounds over the $T$ rounds is at most 
\begin{align*}
\binom{T}{K} p^{K} \leq O\left(\left( \frac{e \cdot T}{K} \cdot \frac{\log^6 n}{\epsilon^8 n}    \right)^K\right) = O\left(\left( \frac{e}{K} \cdot \frac{\log^6 n \cdot (\log \log n + \log(1/\epsilon))}{\epsilon^8  n}    \right)^K\right) = 1/\poly(n)
\end{align*}

Therefore, w.h.p. $M_{T} \leq 2^K$.

Now, given $z \in V$, we will show that $Q_{S_v}(z) - \epsilon \leq Q(z) \leq Q_{S_v}(z) +\epsilon$. For every element $y \in S'_v$, let $Z_y$ be the indicator random variable that denotes whether $x_{t_0(y)} \leq z$. Note that $\Pr(Z_y = 1) = Q(z)$. Therefore, we have
\begin{align*}
\E[\rank_{S_x}(z)] = \sum_{y \in S'_v} \E[Z_y] \cdot m_{S'_v}(y) = Q(z) \cdot |S_v| 
\end{align*}





Since the events $\{Z_y\}_{y \in S_v}$ are mutually independent. We can show that the quantity $\sum_{y \in S_v} Z_y \cdot m_{S'_v}(y)$ concentrates around 
$Q(z) \cdot |S_v(z)|$ by using Hoeffiding's inequality.
\begin{align*}\Pr(|\sum_{y \in S_v} Z_y \cdot m_{S'_v}(y) - Q(z) \cdot |S_v|| \geq \epsilon |S_v|) &\leq \exp\left(-\frac{2\epsilon^2 |S_v|^2 }{\sum_{y \in S_{v}} m_{S'_v}(y)^2} \right) \\
&\leq \exp\left(-\frac{2\epsilon^2 |S_v|^2 }{2^{2K} \cdot |S_v|}  \right) && m_{S'_v}(y) \leq M_T \leq 2^{K}\\
&\leq \exp\left(-\Omega(\log n)\right) &&  |S_v| = \Omega((\log n) /\epsilon^2)\\
&\leq 1/\poly(n)\end{align*}

Therefore, w.h.p.~$|Q_{S_v}(z) - Q(z)| \leq \epsilon$.

\end{proof}

\subsection{Compaction}

In this section, we show that instead of storing the full buffer, it is possible to only keep a summary of the buffer, where the size of the summary is $\Theta(\frac{1}{\epsilon} \cdot (\log \log n + \log (1/\epsilon)) )$. Then, when a node $v$ contacts $t(v)$, $v$ only pulls the summary of $t(v)$. This reduces the message complexity. 

Let $k = \Theta(\frac{1}{\epsilon} \cdot (\log \log n + \log (1/\epsilon)) )$ be a power of 2 that denotes the size of the buffer. Also, we assume that $n' = |S_v|$ is also a power of 2 thus the algorithm runs for $T = \log_2 n' + 1$ rounds. Suppose that $v$ contacts $t(v)$ in a given round, the new buffer updating rule is:
$$\tilde{S}_v \leftarrow \Compact(\tilde{S}_v \cup \tilde{S}_{t(v)})$$, where $\Compact(Z)$ is a {\bf compaction} operation. The compaction operation leaves the buffer unchanged if $|Z| \leq k$. Otherwise, it sort the elements in non-decreasing order and then extract the elements in the even positions. Since $|Z| \leq 2k$, we must have $|\Compact(Z)| \leq k$.


Let $S_{v}$ denote the buffer of $x$ at end of round $T$. Let $\tilde{S}_v \subseteq S_{v}$ denote the buffer of $v$ with the compaction updating rule at the end of round $T$. We will show that $Q_{\tilde{S}_{v}}(z)$ is a good approximation of $Q_{S_{v}}(z)$ for any value $z$.

Suppose that the content of ${S}_v$ is known. At the end of round $i$, $S_v$ can be expressed as the union of the buffers of $2^{T - i}$ nodes. Let $U_v(i)$ denote the multi-set of $2^{T - i}$ nodes. Let $S_{u}(i)$ and $\tilde{S}_{u}(i)$ denote the buffer and the compacted buffer of $u$ at the end of round $i$ respectively.

Now let $U_v(i)$ be the set of nodes such that $\bigcup_{u \in U_v(i)} S_{u}(i) = S_v$, where $|U_v(i)| = 2^{T - i}$. Now we consider the compacted version of the set $C_{v}(i) \defeq \bigcup_{u \in U_v(i)} \tilde{S}_{u}(i)$. Note that $C_{v}(i) = S_v$ for $i = 1,\ldots, \log_2 k + 1$ (that is, before any compaction happens). 	Also, $C_{v}(T) = \tilde{S}_v$. We will show that in each transition along $S_v =  C_v(\log_2k + 1), C_v(\log_2k + 2), \ldots, C_v(T) = \tilde{S}_v$, the error introduced by compaction is sufficiently small. Let $a \cdot S$ denote the multiset where each element $s \in S$ is duplicated to have $a$ occurrence. For a given element $z$, we will compare the rank of $z$ in $w_{i} \cdot C_v(i)$ and $w_{i+1}\cdot C_v(i+1)$ for $\log_2 k + 1 \leq i < T$ to derive the loss of  compaction. 

Let $h_i$ denotes the number of compaction operations done on $S_v$ at the end of round $i$. Therefore, $h_i = \max(0, i - \lg k - 1)$. Let $w_i = 2^{h_i}$ denote the weight of the buffer at the end of round $i$. Note that the cardinality of the set $|w_i \cdot C_v(i)| = |S_v|$ for $1 \leq i \leq T$. 

Given a multi-set $S$ and any element $z$ (not necessarily in $S$), we let $R_S(z)$ to denote the rank of $z$ in $S$ (i.e. the number of elements in $S$ that are smaller than or equal to $z$).

\begin{lemma}\label{lem:dif} Let $v$ be a node and let $z$ by any element. For $\log_2 k + 1 \leq i < T$, $|R_{(w_i \cdot C_v(i))}(z) - R_{(w_{i+1} \cdot C_v(i+1))}(z)| \leq |U_v(i+1)| \cdot w_i $. \end{lemma}

\begin{proof}
Note that $C_v(i+1)$ is obtained from $C_v(i)$ by applying $|U_v(i+1)|$ compactions, where there is one compaction operation for each $u \in U_v(i+1)$. Each compaction operation compacts two buffers $\tilde{S}_u(i)$ and $\tilde{S}_{t(u)}(i)$ into one and also doubles the weight. 

Consider the effect of one compaction operation on the rank of $z$. Let $R_1$ denote the rank of $z$ in $w_i \cdot (\tilde{S}_u(i) \cup \tilde{S}_{t(u)}(i))$ (i.e.~before the compaction) and $R_2$ denote the rank of $z$ in $w_{i+1} \cdot \tilde{S}_{u}(i+1)$ (i.e.~after the compaction). If $R_1$ is even, then we must have $R_1 = R_2$. If $R_1$ is odd, then we have $R_2 = R_1 - w_{i}$. 

Since there are most $|U_{v}(i+1)|$ compactions, the total difference created by all of them is at most $w_{i} \cdot |U_v(i+1)|$.
\end{proof}

\begin{corollary}\label{cor:error_bound} Let $v$ be a node and let $z$ by any element. We have $|R_{S_v}(z) - R_{w_{T} \cdot \tilde{S}_v}(z)| \leq  \frac{n'}{2k} \log \frac{n'}{k}$.\end{corollary}

\begin{proof}
First, note that $S_v = C_v(\log_2 k + 1)$ and $\tilde{S}_v = C_v(T)$. For any $z$, we have:
\begin{align*}
|R_{S_v}(z) - R_{w_{T}\cdot \tilde{S}_v}(z)| &\leq \left|\sum_{i=\log_2 k + 1}^{T-1} R_{w_i \cdot C_v(i)}(z) - R_{w_{i+1} \cdot C_v(i+1)}(z)\right| && w_{\log_2 k + 1} = 1  \\
&\leq  \sum_{i=\log_2 k + 1}^{T-1}\left| R_{w_i \cdot C_v(i)}(z) - R_{w_{i+1} \cdot C_v(i+1)}(z)\right| \\
&\leq \sum_{i=\log_2 k + 1}^{T-1} |U_v(i+1)| \cdot w_i && \mbox{By Lemma \ref{lem:dif}} \\
&= \sum_{i=\log_2 k + 1}^{T-1} 2^{T - i - 1} \cdot 2^{i - \log_2 k - 1} \\
&= \sum_{i=\log_2 k + 1}^{T-1} 2^{T - \log_2 k - 2} = \frac{n'}{2k} \log \frac{n'}{k} 
\end{align*}
\end{proof}

Therefore, if we set $k = \Theta(\frac{1}{\epsilon} \log (\epsilon n'))$, then we have $|R_{S_v}(z) - R_{w_{T} \cdot \tilde{S}_v}(z)| \leq \epsilon n'$ by Corollary \ref{cor:error_bound}. This further implies that $|Q_{S_v}(z) - Q_{\tilde{S}_v}(z)| \leq \epsilon$. We immediately obtain the following Corollary.

\begin{corollary}\label{lem:compaction} Given any $0 < \epsilon < 1$. Let $S_v$ denote the buffer of $v$ with the original updating rule after $T = \Theta(\log n')$ rounds, where $n' = |S_v|$.  Let $\tilde{S}_v$ denote the buffer of $v$ using the compaction updating rule with buffer size $k = \Theta(\frac{1}{\epsilon} \log (\epsilon n'))$. For any $z$, we have $|Q_{S_v}(z) - Q_{\tilde{S}_v}(z)| \leq \epsilon$.\end{corollary}

\begin{theorem}\label{thm:main} Given $\epsilon = \Omega(1/n^{1/16})$. Let $\tilde{S}_v$ denote the buffer of $v$ using the compaction updating rule with buffer size $k = \Theta(\frac{1}{\epsilon} \cdot (\log \log n + \log (1/\epsilon)) )$ after $T = \Theta(\log \log n + \log (1/\epsilon) )$ rounds. For any $z$, w.h.p.~we have $|Q(z) - Q_{ \tilde{S}_v}(z)| \leq \epsilon$.
\end{theorem}
\begin{proof}


Let $\epsilon_1 = \epsilon_2 = \epsilon/ 2$. First, by applying Lemma  \ref{lem:double_sample}, we have $|Q_{S_v}(z) - Q(z) | \leq \epsilon_1$ w.h.p. where $n' = |S_v| = \Theta(\log n / \epsilon^2)$. Then, we apply Corollary \ref{lem:compaction} to show that the corresponding compaction update rule obtains a compacted buffer $\tilde{S}_v$ of size $O(\frac{1}{\epsilon}(\log \log n + \log (1/\epsilon)))$ such that $|Q_{S_v}(z) - Q_{\tilde{S}_v}(z)| \leq \epsilon_2$. By triangle inequality, we have:
$$|Q(z) - Q_{\tilde{S}_v}(z)| \leq |Q_{S_v}(z) - Q(z) |+|Q_{S_v}(z) - Q_{\tilde{S}_v}(z)| \leq \epsilon_1 + \epsilon_2 \leq \epsilon \qedhere
$$ \end{proof}

\section{Tools}
\begin{lemma}\label{lem:chernofforignial} \textnormal{(Chernoff Bound)} Let $X_1,\dots, X_n$ be indicator variables such that $\Pr(X_i=1) = p$. Let $X = \sum_{i=1}^{n} X_i$. Then, for $\delta > 0$:
\begin{align*}
&&\Pr(X\geq(1+\delta)\E[X]) &< \left[ \frac{e^{\delta}}{(1+\delta)^{(1+\delta)}} \right]^{\E[X]}\\
&&\Pr(X\leq(1-\delta)\E[X]) &< \left[ \frac{e^{\delta}}{(1-\delta)^{(1-\delta)}} \right]^{\E[X]}
\end{align*}
The two bounds above imply that for $0 < \delta < 1$, we have:
\begin{align*}
&&\Pr(X\geq(1+\delta)\E[X]) &< e^{-\delta^2 \E[X] / 3} \\
&&\Pr(X\leq(1-\delta)\E[X]) &< e^{-\delta^2 \E[X] / 2}.
\end{align*}
\end{lemma}

\begin{lemma}\label{lem:largedev} 
Suppose that for any $\delta > 0$, 
$$\Pr\left(X  >(1+\delta)np \right) \leq \left[ \frac{e^{\delta}}{(1+\delta)^{(1+\delta)}} \right]^{np}$$ then for any $M \geq np$ and $0 < \delta < 1$,
\begin{align*} \Pr\left(X  > np + \delta M \right) &\leq \left[ \frac{e^{\delta}}{(1+\delta)^{(1+\delta)}} \right]^{M} \\ &\leq e^{-\delta^2M/3}\end{align*}
\end{lemma}
\begin{proof}
Without loss of generality, assume $M = tnp$ for some $t \geq 1$, we have 
\begin{align*}
&&&\Pr\left(X > np + \delta M\right) \\
&&&\leq \left[ \frac{e^{t \delta}}{(1+t \delta)^{(1+t \delta)}} \right]^{np} \\
&&&=  \left[ \frac{e^{\delta}}{(1+t \delta)^{(1+t \delta)/t}} \right]^{M} \\
&&&\leq \left[ \frac{e^{\delta}}{(1+\delta)^{(1+\delta)}} \right]^{M} \tag*{$(*)$} \\
&&& \leq e^{-\delta^2 M/ 3} \tag*{$\frac{e^{\delta}}{(1+\delta)^{(1+\delta)}} \leq e^{-\delta^2 / 3}$ for $0 < \delta < 1$}
\end{align*}
Inequality (*) follows if $(1+t \delta)^{(1+t \delta)/t} \geq (1+\delta)^{(1+\delta)}$, or equivalently, $((1+t \delta)/t) \ln (1+t \delta) \geq (1+\delta) \ln (1+\delta)$. 
Letting $f(t) = ((1+t \delta)/t )\ln (1+t \delta) - (1+\delta) \ln (1+\delta)$, we have $f'(t) = \frac{1}{t^2}\left(\delta t - \ln(1+\delta t) \right) \geq 0$ for $t > 0$. Since $f(1) = 0$ and $f'(t) \geq 0$ for $t > 0$, we must have $f(t) \geq 0$ for $t\geq 1$. 
\end{proof}

\end{document}